\documentclass[a4paper,11pt]{article}
\usepackage{amssymb,amsmath,amsthm,amsfonts}
\usepackage{xspace}
\usepackage[ruled,vlined,linesnumbered]{algorithm2e}
\usepackage{mathrsfs}  
\usepackage{cite}

 \usepackage[pdftex, plainpages = false, pdfpagelabels, 
                 bookmarks=false,
                 bookmarksopen = true,
                 bookmarksnumbered = true,
                 breaklinks = true,
                 linktocpage,
                 pagebackref,
                 colorlinks = true,  
                 linkcolor = blue,
                 urlcolor  = blue,
                 citecolor = red,
                 anchorcolor = green,
                 hyperindex = true,
                 hyperfigures
                 ]{hyperref}

\usepackage[obeyFinal,colorinlistoftodos]{todonotes}

\newcommand{\pname}{\textsc}
\newcommand{\ProblemFormat}[1]{\pname{#1}}
\newcommand{\ProblemIndex}[1]{\index{problem!\ProblemFormat{#1}}}
\newcommand{\ProblemName}[1]{\ProblemFormat{#1}\ProblemIndex{#1}{}\xspace}

\newcommand{\probClust}{\ProblemName{Binary $k$-Means}}
\newcommand{\probAtMostClust}{\ProblemName{Binary $k$-Means}}
\newcommand{\probProjCl}{\ProblemName{Binary Projective Clustering}}

\newcommand{\probBFact}{\ProblemName{Low Boolean-Rank Approximation}}
\newcommand{\probFact}{\ProblemName{Low GF(2)-Rank Approximation}}

 \newcommand{\bmfgfr}{\ProblemName{Low GF(2)-Rank Approximation}}

\newcommand{\rclustering}{{\sc Binary Constrained Clustering}\xspace}

\newcommand{\bmfbr}{{\sc Low Boolean-Rank Approximation}\xspace}

\newcommand{\kmclust}{{\sc   $k$-Means Clustering}\xspace}

\newcommand{\GF}{{GF}(2)\xspace}
\newcommand{\argmin}{{\rm argmin}\xspace}
\newcommand{\rank}{{\rm rank}\xspace}
\newcommand{\GFrank}{{\rm{GF}}(2)\text{{\rm -rank}}\xspace}
\newcommand{\hdist}{d_H}
\newcommand{\Hdist}{{\operatorname{cost}}}
\newcommand{\opt}{{\sf OPT}}
\newcommand{\bfA}{\mathbf{A}} 
\newcommand{\bfB}{\mathbf{B}}

\newcommand{\bfU}{\mathbf{U}}
\newcommand{\bfV}{\mathbf{V}}  

\newcommand{\bfa}{\mathbf{a}} 
\newcommand{\bfb}{\mathbf{b}} 
\newcommand{\bfc}{\mathbf{c}}

\newcommand{\bfp}{\mathbf{p}} 
\newcommand{\bfs}{\mathbf{s}} 
\newcommand{\bft}{\mathbf{t}}
\newcommand{\bfu}{\mathbf{u}} 
\newcommand{\bfv}{\mathbf{v}} 
\newcommand{\bfw}{\mathbf{w}} 
\newcommand{\bfrho}{\mathbf{q}} 
 
\newcommand{\bfQ}{\mathbf{Q}} 
\newcommand{\bfx}{\mathbf{x}}

\newcommand{\bfy}{\mathbf{y}} 

\newcommand{\A}{\mathcal{A}} 
\newcommand{\cR}{\mathcal{R}}

\newtheorem{definition}{Definition}
\newtheorem{theorem}{Theorem}
\newtheorem{claim}{Claim}
\newtheorem{corollary}{Corollary}
\newtheorem{lemma}{Lemma}
\newtheorem{observation}[theorem]{Observation}
\newtheorem{proposition}[theorem]{Proposition}

\newcommand{\OO}{{\cal O}}

\newcommand{\BB}{{\cal B}}
\newcommand{\RR}{{\cal R}}
\newcommand{\PP}{{\cal P}}

 \newcommand{\proj}[2]{{\sf proj}_{#2}(#1)}
\newcommand{\rest}[3]{{#1}|_{(#2,#3)}}
\newcommand{\ar}[2]{r}
\newcommand{\best}[1]{{\sf best}_{#1}}

\newcommand{\defproblem}[3]{
  \vspace{3mm}
\noindent\fbox{
  \begin{minipage}{.95\textwidth}
  \begin{tabular*}{\textwidth}{@{\extracolsep{\fill}}lr} \textsc{#1} \\ \end{tabular*}
  {\bf{Input:}} #2  \\
  {\bf{Task:}} #3
  \end{minipage}
  }
  \vspace{2mm}
}

\newcommand{\expect}{\mathbf{E}}
\newcommand{\expectation}[3][0]{%
   \ifcase#1 
        \expect [ #2 \mid #3 ] 
      \or   \expect\bigl[ #2 \bigm\vert #3 \bigr]
      \or   \expect \Bigl[ #2 \Bigm\vert #3 \Bigr] 
      \or \expect\biggl[ #2 \biggm\vert #3 \biggr]
      \or  \expect \Biggl[#2 \Biggm\vert #3 \Biggr]
   \else 
       \expect\left[ #2  \;\middle\vert\; #3 \right]
   \fi}

\newcommand{\probability}[3][0]{%
   \ifcase#1 
    \Pr( #2 \mid #3 ) 
      \or \Pr \bigl( #2 \bigm\vert #3 \bigr)
      \or \Pr \Bigl( #2 \Bigm\vert #3 \Bigr) 
      \or \Pr \biggl( #2 \biggm\vert #3 \biggr)
      \or \Pr \Biggl( #2 \Biggm\vert #3 \Biggr) 
   \else 
      \Pr \left( #2  \;\middle\vert\; #3 \right)
   \fi}

\usepackage{vmargin}
\setmarginsrb{1in}{1in}{1in}{1in}{0pt}{0pt}{0pt}{6mm}

\title{Approximation Schemes for    Low-Rank  Binary Matrix Approximation  Problems\thanks{The research leading to these results have  been supported by the Research Council of Norway via the projects ``CLASSIS'' and ``MULTIVAL".}
}
\author{
Fedor V. Fomin\thanks{
Department of Informatics, University of Bergen, Norway.} \addtocounter{footnote}{-1}
\and
Petr A. Golovach\footnotemark{} \addtocounter{footnote}{-1}
\and
Daniel Lokshtanov\footnotemark{} \addtocounter{footnote}{-1}
\and 
Fahad Panolan\footnotemark{}\addtocounter{footnote}{-1}
\and 
Saket Saurabh\footnotemark{}
}

\date{}

\begin{document}
\maketitle

\thispagestyle{empty}

%
%


\begin{abstract} 

We provide a randomized  linear time approximation scheme for a generic problem about  clustering  of binary vectors subject to additional constrains. The new constrained clustering problem 
encompasses a number of  problems and by solving it,  we obtain  the first  linear time-approximation schemes for a number of well-studied fundamental problems concerning clustering of binary vectors and low-rank approximation of binary matrices. 
Among the problems solvable by our approach are  \bmfgfr,   \bmfbr, and various versions of \textsc{Binary Clustering}.
For example, for \bmfgfr problem, where  for an $m\times n$ binary matrix $A$ and integer $r>0$, we seek for a binary matrix $B$ of \GF rank at most $r$ such that $\ell_0$ norm of matrix
$\bfA-\bfB$ is minimum, our algorithm, for any  $\epsilon>0$ 
in time $ f(r,\epsilon)\cdot n\cdot m$, where $f$ is some computable function,   outputs a $(1+\epsilon)$-approximate solution with probability at least $(1-\frac{1}{e})$. Our approximation algorithms substantially improve the running times and approximation factors of previous works. We also give (deterministic) PTASes for these problems 
running in time $n^{f(r)\frac{1}{\epsilon^2}\log \frac{1}{\epsilon}}$, where $f$ is some function depending on the problem.  Our algorithm for the constrained clustering problem is based on a novel sampling lemma, which is interesting in its own. 

\end{abstract}

\newpage
\pagestyle{plain}
\setcounter{page}{1}

\section{Introduction}

We define a new clustering problem which encompasses a number of well studied problems about low-rank approximation of binary matrices and clustering of binary vectors.

In order to obtain approximation algorithms for low-rank approximation problem, we design approximation algorithms for a ``constrained'' version of binary clustering. 
 
A $k$-ary relation $R$ is a set of binary $k$-tuples with elements from $\{0,1\}$. A $k$-tuple $t=(t_1,\dots, t_k)$ \emph{satisfies} $R$, we write $t\in R$, if $t$ is equal to one of $k$-tuples from $R$.  
\begin{definition}[Vectors satisfying $\cR$]
Let $\cR=\{R_1, \dots, R_m\}$ be a set of $k$-ary relations. We say that a set $C=\{\bfc_1, \bfc_2, \dots, \bfc_k\}$ of binary $m$-dimensional vectors  \emph{satisfies $\cR$} and write $<C,\cR>$, if 
 $(\bfc_1[i],\ldots,\bfc_k[i])\in R_i$ for all $i\in \{1,\ldots,m\}$.
\end{definition}

For example, for $m=2$, $k=3$, $R_1=\{(0,0,1), (1,0,0)\}$, and  $R_2=\{(1,1,1), (1,0,1), (0,0,1)\}$, the set of vectors 
\[
\bfc_1=\left(
\begin{array}{c}
0\\
1\\
\end{array}
\right) , \, 
\bfc_2=\left(
\begin{array}{c}
0\\
0\\
\end{array}
\right) , \,
\bfc_3=\left(
\begin{array}{c}
1\\
1\\
\end{array}
\right)   \]
satisfies $\cR=\{R_1,  R_2\}$ because  $(\bfc_1[1],\bfc_2[1],\bfc_3[1])=(0,0,1)\in   R_1$ and $({\bfc}_1[2],\bfc_2[2],\bfc_3[2])=(1,0,1)\in   R_2$.

Let us recall that the \emph{Hamming distance} between two vectors $\bfx, \bfy\in\{0,1\}^m$, where $\bfx=(x_1,\ldots,x_m)^\intercal$ and $\bfy=(y_1,\ldots,y_m)^\intercal$, is $\hdist(\bfx,\bfy)=\sum_{i=1}^m |x_i-y_i|$ or, in   words, the number of positions $i\in\{1,\ldots,m\}$ where $x_i$ and $y_i$ differ. For a set of vectors $C$ and vector $x$, we define $\hdist(\bfx,C)$, the Hamming distance   between 
$x$ and $C$, as the minimum Hamming distance between $x$ and a vector from $C$. Thus 
$\hdist(\bfx,C)=\min_{\bfc\in C}\hdist(\bfx,\bfc)$.


Then we define the following problem.

\defproblem{\rclustering}{A set $X\subseteq \{0,1\}^m$ of $n$ vectors, a positive integer $k$ and a set of $k$-ary relations
$\cR=\{R_1, \dots, R_m\}$. }{Among all  vector sets $C=\{\bfc_1,\ldots,\bfc_k\}\subseteq \{0,1\}^m$ satisfying $\cR$, find a set $C$ minimizing the sum 
$\sum_{\bfx\in X} \hdist(\bfx,C)$.}

 
First we prove the following theorem. 

\begin{theorem}
\label{thm:mainrclusteringdet}
There is a deterministic algorithm  which  given instance of \rclustering\ and $\epsilon>0$, runs in time $ m\cdot n^{\OO(\frac{k^2}{\epsilon^2}\log \frac{1}{\epsilon})}$  and outputs a $(1+\epsilon)$-approximate solution.   
\end{theorem}

Our main result is the following theorem. 
\begin{theorem}
\label{thm:mainrclustering}
There is an algorithm  which for a given instance of \rclustering\ and $\epsilon>0$ in time 
$2^{\OO\left(\frac{ k^4}{\varepsilon^2} \log \frac{1}{\varepsilon}\right)}\cdot \left(\frac{1}{\varepsilon}\right)^{\OO\left(\frac{ k}{\varepsilon^2} \log \frac{1}{\varepsilon}\right)}n \cdot m$
 outputs a $(1+\epsilon)$-approximate solution with probability at least $(1-\frac{1}{e})$. 
\end{theorem}
In other words, the algorithm outputs a set of vectors $C=\{\bfc_1,\ldots,\bfc_k\}\subseteq \{0,1\}^m$ satisfying $\cR$ such that  $\sum_{\bfx\in X} \hdist(\bfx,C)\leq (1+\epsilon)\cdot OPT$, where $OPT$ is the value of the optimal solution.

Theorems~\ref{thm:mainrclusteringdet} and \ref{thm:mainrclustering} have a number of interesting applications.


\subsection{Applications of the main theorem}
Binary matrix factorization is the following problem.  Given a binary  $m\times n$ matrix, that is a matrix with entries from domain $\{0,1\}$, 
\begin{equation*}\bfA=
\begin{pmatrix}
a_{11}&a_{12}& \ldots& a_{1n}\\
a_{21}&a_{21}& \ldots& a_{2n}\\
\vdots & \vdots & \ddots & \vdots\\ 
a_{m1}&a_{m2} &\ldots &a_{mn}
\end{pmatrix}=(a_{ij})\in \{0,1\}^{m\times n}, 
\end{equation*}
the task is to find a ``simple'' binary  $m\times n$ matrix $\bfB$  which  approximates $\bfA$ subject to some specified constrains. One of the   most widely studied error measures is the \emph{Frobenius norm}, which for the  matrix $\bfA$ is defined as
\begin{equation*}
\|\bfA\|_F=\sqrt{\sum_{i=1}^m\sum_{j=1}^n|a_{ij}|^2}.
\end{equation*}
Here the sums are taken over $\mathbb{R}$. Then we want to find  a matrix $\bfB$ with certain properties such that 
\[\|\bfA-\bfB\|_F^2\] is minimum.

In particular, two variants of the problem were studied in the literature, in the first variant on seeks for a matrix $\bfB$ of small \GF-rank. In the second variant, matrix $\bfB$ should be of Boolean rank $r$. Depending on the selection of the rank, we obtain two different optimization problems.

 \paragraph{Low GF(2)-Rank Approximation.} 
 Here 
 the task is to approximate a given binary matrix $\bfA$ by $\bfB$ that has \GF-rank $r$.

\defproblem{\bmfgfr}%
{An $m\times n$-matrix $\bfA$ over \GF and a positive integer $r$.}%
{Find a binary  $m\times n$-matrix $\bfB$   with \GFrank$(\bfB)\leq r$ such that $\|\bfA-\bfB\|_F^2$ is minimum.
}

\paragraph{Low Boolean-Rank Approximation.}
Let $\bfA$ be a binary $m\times n$ matrix. Now we consider the elements of $\bfA$ to be \emph{Boolean} variables. 
The \emph{Boolean rank} of $\bfA$ is the minimum $r$ such that $\bfA=\bfU\wedge \bfV$ for a Boolean $m\times r$ matrix $\bfU$ and a Boolean $r\times n$ matrix $\bfV$, where the product is Boolean, that is,  the logical $\wedge$ plays the role of multiplication and $\vee$ the role of sum. Here  $0\wedge 0=0$, $0 \wedge 1=0$, $1\wedge 1=1$ , $0\vee0=0$, $0\vee1=1$, and  $1\vee 1=1$.  
Thus the  matrix product is over the Boolean semi-ring $({0, 1}, \wedge, \vee)$. This can be equivalently expressed
as the  normal matrix product with addition defined as $1 + 1 =1$. Binary matrices equipped with such algebra are called \emph{Boolean
matrices}.

\defproblem{\bmfbr}%
{A Boolean $m\times n$ matrix $\bfA$ and a positive integer $r$.}%
{Find  a Boolean $m\times n$ matrix $\bfB$ of Boolean rank at most $r$ such that  $\|\bfA-\bfB\|_F^2$ is minimum. 
}

Low-rank matrix approximation problems can be also treated as special cases of \rclustering.

\begin{lemma}\label{lem:matrixFasRClust} For any instance $(\bfA,r)$ of  \bmfgfr, 
 one can  construct in time $\OO(m+n+2^{2r})$ an instance $(X,k=2^r,\cR)$ of \rclustering\ 
with the following property. 
Given any $\alpha$-approximate solution $C$ of $(X,k, \cR)$,  an $\alpha$-approximate solution $\bfB$ of $(\bfA,r)$ can be constructed in time $\OO(rmn)$ and vice versa.
%
%
\end{lemma}


\begin{proof}[Proof Sketch]
Observe that if $\GFrank(\bfB)\leq r$, then $\bfB$ has at most $2^r$ pairwise distinct columns, because each column is a linear combination of at most $r$ vectors of a basis of the column space of $\bfB$.  Also the task of \bmfgfr{} can equivalently be stated as follows: find vectors $\bfs_1,\ldots,\bfs_r\in\{0,1\}^m$ over \GF{} such that 
$$\sum_{i=1}^n\min\{\hdist(\bfs,\bfa_i)\mid \bfs~\text{is a linear combination of}~\bfs_1,\ldots,\bfs_r\}$$ is minimum, where
$\bfa_1,\ldots,\bfa_n$ are the columns of $\bfA$. 
Respectively, to  encode an instance of \bmfgfr\ as an instance of \rclustering, 
we construct the following relation $R$. Set $k=2^r$.
Let $\Lambda=(\mathbf{\lambda}_1,\ldots,\mathbf{\lambda}_k)$ be the $k$-tuple composed by  all pairwise distinct vectors of $\{0,1\}^r$.  Thus each element $\lambda_i\in \Lambda$ is a binary $r$-vector. We define 
 $R=\{(x^\intercal \mathbf{\lambda}_1,\ldots,x^\intercal\mathbf{\lambda}_k)\mid x\in\{0,1\}^r\}.$  Thus $R$ consists of $2^r$ $k$-tuples and every $k$-tuple in $R$ is a row of the matrix $\Lambda^\intercal \cdot \Lambda$.
Now we define $X$ to be the set of columns of $\bfA$ and for each $i\in [m]$, $R_i=R$. Note that since all $R_i$ are  equal, we can construct and keep just one copy of $R$. 

One can show that 
if $B$ is a solution to $(\bfA,r)$, then all linear combinations $C$ of a basis of the column vectors of $\bfB$ is a solution to $(X,k,\RR)$ and 
the cost of $C$ is at most the cost of the solution $\bfB$ of $(\bfA,r)$. 
Similarly,   if $C$ is a solution to $(X,k,\RR)$, then solution $\bfB$ to $(\bfA,r)$ is constructed from $C$  by taking the the $j$-{th} column of $\bfB$ be 
equal to the vector in $C$ which is closest to the $j$-{th} column vector of $\bfA$. 
Clearly the  cost of $\bfB$ is at most the cost of $C$.  It is easy to see that given $\bfB$, one can construct 
$C$ in $\OO(rmn)$ time and vice versa. 
\end{proof}

%

Similarly  for \bmfbr we have the following lemma. 
\begin{lemma}\label{lem:BmatrixFasRClust} For any instance $(\bfA,r)$ of  \bmfbr, 
one can  construct in time $\OO(m+n+2^{2r})$ an instance $(X,k=2^r,\cR)$ of \rclustering\ with the following property. 
Given any $\alpha$-approximate solution $C$ of $(X,k, \cR)$ an $\alpha$-approximate solution $\bfB$ of $(\bfA,r)$ can be constructed in time $\OO(rmn)$ and vice versa.
\end{lemma}


The proof essentially repeats the proof of Lemma~\ref{lem:matrixFasRClust}. We are now working with the Boolean semi-ring $({0, 1}, \wedge, \vee)$ but still we can use exactly the same trick to  reduce  \bmfbr{} to \rclustering. The only difference is that  \GF{} summations and products are replaced by $\vee$ and $\wedge$ respectively in the definition of the relation $R$.  Thus every   $k$-tuple in $R$ is a row of the matrix $\Lambda^\intercal \wedge \Lambda$.

Hence  to design approximation schemes for 
  \bmfbr and \bmfgfr, it suffice  to give an approximation scheme for \rclustering. The main technical contribution  of 
the paper is the proof of the following theorem.



 For  $\alpha>1$, we say that an algorithm is an $\alpha$-approximation algorithm for
the low-rank approximation problem if for a matrix $A$ and an integer $r$ it  outputs a matrix $B$ satisfying the required constrains such that $\|\bfA -\bfB\|_F^2\leq \alpha \cdot \|\bfA -\bfB_r\|_F^2$, where
 $\bfB_r= \argmin_{\rank(\bfB_r)=r} \|\bfA -\bfB_r\|_F^2$.
By Theorems~\ref{thm:mainrclusteringdet} and \ref{thm:mainrclustering} and Lemmata~\ref{lem:matrixFasRClust} and~\ref{lem:BmatrixFasRClust}, we obtain the following.
\begin{corollary}\label{thm:matrixfactor}
$(i)$ There is a deterministic algorithm  which for a given instance of \bmfbr (\bmfgfr) and $\epsilon>0$ in time  $ m\cdot n^{\OO(\frac{2^{2r}}{\epsilon^2}\log \frac{1}{\epsilon})}$ outputs a $(1+\epsilon)$-approximate solution.
$(ii)$ There is an algorithm  which for a given instance of \bmfbr (\bmfgfr) and $\epsilon>0$ in time 
$  \left(\frac{1}{\epsilon} \right)^{\left(\frac{2^{\OO(r)}}{\epsilon^2}\log \frac{1}{\epsilon}\right)}\cdot n\cdot m$ 
outputs a $(1+\epsilon)$-approximate solution with probability at least $(1-\frac{1}{e})$.
\end{corollary}


Let us observe that our results also yield randomized approximation scheme for the ``dual'' maximization versions of the low-rank matrix approximation problems. In these problems one 
wants to  maximize the number of elements that are the same in $A$ and $B$ or, in other words, to maximize the value of $nm-\|\bfA-\bfB\|_F^2$. 
 It is easy to see that for every binary $m\times n$ matrix $\bfA$ there is a binary matrix $\bfB$ with  \GFrank$(\bfB)\leq 1$ such that $\|\bfA-\bfB\|_F^2\leq mn/2$.
 This observation implies that  
 \begin{align*}
(mn-\|\bfA-\bfB^*\|_F^2)-(mn-\|\bfA-\bfB\|_F^2)=&\|\bfA-\bfB\|_F^2-\|\bfA-\bfB^*\|_F^2\leq\\
 \varepsilon \|\bfA-\bfB^*\|_F^2\leq &\varepsilon (mn-\|\bfA-\bfB^*\|_F^2).
\end{align*}

\subsection{Binary clustering and variants} The special case of 
\rclustering where no constrains are imposed on the centers of the clusters is \probClust.


\defproblem{\probClust}{A set $X\subseteq \{0,1\}^m$ of $n$ vectors and  a positive integer $k$. }{Find a set  $C=\{\bfc_1,\ldots,\bfc_k\}\subseteq \{0,1\}^m$  minimizing the sum 
$\sum_{\bfx\in X} \hdist(\bfx,C)$.}

Equivalently, in \probClust we seek to partition a  set of binary vectors $X$ into $k$ clusters $\{X_1,\ldots,X_{k}\}$ such that after  we  assign to each cluster its mean, which is a binary vector $c_i$ (not necessarily  from $X$)  closest to $X_i$, then the sum  
$\sum_{i=1}^{k}\sum_{x\in X_i}\hdist(\bfc_i,\bfx) $ is minimum. 

Of course,  \rclustering generalizes \probClust:  For given instance $(X,k)$ of \probClust    by taking sets $R_i$, $1\leq i\leq m$,  consisting of all possible $k$-tuples   $\{0,1\}^k$, we construct in time $\OO(n+m+ 2^k)$ an  instance $(X,k,\cR)$ of  \rclustering equivalent to $(X,k)$. 
Note that since all the sets $R_i$ are the same, it is sufficient to keep just one copy of the set for the instance $(X,k,\cR)$.  
That is, any $(1+\varepsilon)$-approximation to one instance is also a   $(1+\varepsilon)$-approximation to another.  Theorems~\ref{thm:mainrclusteringdet} and \ref{thm:mainrclustering} implies that 


\begin{corollary}\label{thm:binarycluster}
$(i)$ There is a deterministic algorithm  which for a given instance of \probClust and $\epsilon>0$ in time  $ m\cdot n^{\OO(\frac{k^2}{\epsilon^2}\log \frac{1}{\epsilon})}$ outputs a $(1+\epsilon)$-approximate solution.
$(ii)$ There is an algorithm  which for a given instance of  \probClust  and $\epsilon>0$ in time $ \left(\frac{1}{\epsilon} \right)^{\OO\left(\frac{k^4}{\epsilon^2}\log \frac{1}{\epsilon}\right)}\cdot n\cdot m$  outputs a $(1+\epsilon)$-approximate solution with probability at least $(1-\frac{1}{e})$. 
\end{corollary}

  Theorems~\ref{thm:mainrclusteringdet} and \ref{thm:mainrclustering} can be used for many other variants of binary clustering.  Let us briefly mention some other clustering problems which fit in our framework.
  
For example, the following generalization  of binary clustering can be formulated as \rclustering.
Here the centers of clusters are linear subspaces of bounded dimension $r$. (For $r=1$ this is \rclustering and for $c=1$ this is \bmfgfr.)
More precisely, in 
\probProjCl we are given a set $X\subseteq \{0,1\}^m$ of $n$ vectors and  positive integers $k$ and $r$. The task is to find a family of $r$-dimensional linear subspaces   $C=\{C_1,\ldots,C_k\} $  over \GF minimizing the sum 
\[\sum_{\bfx\in X} \hdist(\bfx,\cup_{i=1}^kC).\]

%

To see that \probProjCl is the special case of  \rclustering, we observe that the condition that $C_i$ is a $r$-dimensional subspace over \GF can be encoded (as in Lemma~\ref{lem:matrixFasRClust}) by $2^r$ constrains.  Similar arguments hold also for the variant of  \probProjCl when instead of $r$-dimensional subspaces we use $r$-flats
 ($r$-dimensional affine subspaces).

  In \textsc{Correlative $k$-Bicluster Editing}, we are given a bipartite graph and the task is to change the minimum number of adjacencies such that the resulting graph is the disjoint union of at most $k$ complete bipartite graphs \cite{amit2004bicluster}.  This is the special case of   \rclustering where each constrain $R_i$ consists of $k$-tuples and each of the $k$-tuples contains exactly one element $1$ and all other elements $0$. 
   Another problem which  can be reduced to  \rclustering is   the following  variant of the {\sc Biclustering} problem \cite{WulffUB13}.  
  Here  for matrix $\bfA$, and positive integers $k,r$, we want to find a binary  $m\times n$-matrix $\bfB$  such that $\bfB$ has at most $r$ pairwise-distinct rows, $k$ pairwise distinct-columns
such that 
   $\|\bfA-\bfB\|_F^2$ is minimum.

\subsection{Previous work}
\paragraph{Low-rank binary matrix approximation}

Low-rank matrix approximation   is a fundamental and extremely well-studied problem. 
 When the measure of the similarity between $\bfA$ and $\bfB$  is the Frobenius norm of matrix $\bfA-\bfB$, the 
rank-$r$ approximation (for any $r$) of matrix $\bfA$ can be efficiently found via the singular value decomposition (SVD).  
This is an extremely well-studied problem and we refer to  surveys and books  \cite{RavindrV08,MahonyM11,Woodr14} for an overview of this topic. 
However, SVD does not guarantee to find an optimal solution in the case when additional structural constrains on  the low-rank approximation matrix $\bfB$ (like being non-negative or  binary) are imposed.
In fact, most of the  variants of low-rank approximation with additional constraints are NP-hard.

For long time  the predominant approaches for solving  such low-rank approximation problems with NP-hard  constrains were either heuristic methods based on convex relaxations or optimization methods.  
Recently, there has been considerable interest
in the  rigorous   analysis  of such problems   \cite{AroraGKM12,ClarksonW15,Moitra16,RazenshteynSW16}.

 \paragraph{\bmfgfr} 
 arises naturally in applications involving binary data sets and 
serve as  important tools in dimension reduction for high-dimensional data sets with binary attributes, see  \cite{DanHJWZ15,Jiang2014,GutchGYT12,Koyuturk2003,PainskyRF16,Shen2009,Yeredor11} for further references. Due to the numerous applications
of  low-rank binary 
matrix approximation, various heuristic algorithms for these problems could be found in the literature
 \cite{DBLP:conf/icdm/JiangH13a,Jiang2014,fu2010binary,Koyuturk2003,Shen2009}. 

When it concerns a rigorous analysis of  algorithms for  \bmfgfr, the previous results  include the following.  
Gillis and Vavasis ~\cite{GillisV15} and Dan et al. \cite{DanHJWZ15} have shown that \bmfgfr is NP-complete  for every $r\geq1$. A subset of the authors studied parameterized algorithms for \bmfgfr in \cite{abs-1803-06102}.

The first approximation algorithm for \bmfgfr is due to 
Shen et al. \cite{Shen2009} who gave a $2$-approximation algorithm for the special case of $r=1$. Shen et al. \cite{Shen2009} formulated the rank-one problem as Integer Linear Programming and proved that its relaxation gives a 2-approximation. 
They also observed that the  efficiency of their algorithm can be improved by reducing the linear program to the Max-Flow problem. Jiang et al. \cite{Jiang2014} found a much simpler algorithm by observing that for the rank-one case, simply selecting the best column of the input matrix yields a 2-approximation. Bringmann et al.  \cite{ BringmannKW17} developed a  2-approximation algorithm for $r=1$ which runs in sublinear time. Thus even for the special case $r=1$ no polynomial time approximation scheme was known prior to our work.

For rank $r>1$,  Dan et al. \cite{DanHJWZ15} have shown that a $(r/2 +1 +\frac{r}{2(2^r-1)})$-approximate solution can be formed from $r$ columns of the input matrix $\bfA$. Hence by trying all possible $r$ columns of $\bfA$,  we can obtain    $r/2 +1 +\frac{r}{2(2^r-1)}$-approximation in time $n^{\OO(r)}$.  Even the  existence of a   linear time algorithm with a constant-factor approximation for $r>1$ was open. 

%
%

\paragraph{\probBFact}
 in case of 
  $r=1$ coincides with \probFact. Thus   by the results of Gillis and Vavasis ~\cite{GillisV15} and Dan et al. \cite{DanHJWZ15} \probBFact is  NP-complete already for $r=1$. 
 While computing  \GF-rank (or rank over any other field) of a matrix can be performed in polynomial time, 
deciding whether the Boolean rank of a given matrix is at most $r$ is already an NP-complete problem. 
 This follows from the well-known relation between the Boolean rank and covering edges of a bipartite graph by bicliques \cite{GregoryPJL91}. 
Thus for fixed $r$, the problem is solvable in time $2^{2^{\OO(r)}}(nm)^{\OO(1)}$ ~\cite{GrammGHN08,FominGP18} 
  and unless Exponential Time Hypothesis (ETH) fails, it cannot be solved in time $2^{2^{o(r)}}(nm)^{\OO(1)}$ \cite{ChandranIK16}.

  There is a large body of work on \probBFact, especially in the data mining and knowledge discovery communities. 
In data mining, matrix decompositions are often used to produce concise representations of data. 
Since much of the real data such as word-document data  is binary or even Boolean in nature, 
Boolean low-rank approximation could provide a deeper insight into 
the semantics associated with the original matrix. There is a big body of work done on \probBFact, see e.g.  \cite{Bartl2010,BelohlavekV10,DanHJWZ15,LuVAH12,MiettinenMGDM08,DBLP:conf/kdd/MiettinenV11,DBLP:conf/icde/Vaidya12}. In the literature the problem appears under different names like \textsc{Discrete Basis Problem} \cite{MiettinenMGDM08} or 
\textsc{Minimal Noise Role Mining Problem} \cite{VaidyaAG07,LuVAH12,Mitra:2016}.

Since for $r=1$ \probBFact is equivalent to  \probFact, the 2-approximation algorithm for \probFact in the case of $r=1$ is also a 2-approximation algorithm for  \probBFact. 
For rank $r>1$,  Dan et al. \cite{DanHJWZ15} described a procedure which produces a  $2^{r-1} +1$-approximate solution to the problem. 

\textbf{Let us note  that independently   Ban et al.~\cite{2018arXiv180706101B} obtained a very similar algorithmic result for low-rank binary approximation. Their algorithm runs in time $\left(\frac{1}{\epsilon}\right)^{2^{\OO(r)}/\epsilon^2} n \cdot m \cdot \log^{2^r} n$. Moreover, they also obtained a lower bound of $2^{2^{\delta r}}$ for a constant $\delta$ under  Small Set Expansion Hypothesis and Exponential Time Hypothesis.   Surprisingly, at first glance, the technique and approach in \cite{2018arXiv180706101B} 
to obtain algorithmic result for low-rank binary approximation is similar to that of ours. 
}

\paragraph{\probClust} was introduced by Kleinberg, Papadimitriou,  and Raghavan \cite{KleinbergPR04} as one of the examples of segmentation problems.  Ostrovsky  and Rabani \cite{OstrovskyR02} gave a randomized PTAS for  \probClust. In other words they show that   for any $\gamma>0$ and $0<\varepsilon< 1/8$ there is an  algorithm 
finding an $(1+8\varepsilon)^2$-approximate solution with probability at least $1-n^{-\gamma}$. The 
  running time of the algorithm of  Ostrovsky  and Rabani  is $n^{f(\varepsilon, k)}$ for some function $f$.
  No Efficient Polynomial Time Approximation Scheme (EPTAS), i.e. of running time ${f(\varepsilon, k)}\cdot n^c$, for this problem  was known prior to our work. 
  
   For the dual maximization problem, where one wants to  maximize  $nm-\sum_{i=1}^{k}\sum_{x\in X_i}\hdist(c_i,x) $ a significantly faster approximation is known. 
Alon and Sudakov \cite{AlonS99} gave a randomized EPTAS.
For a fixed $k$ and $\varepsilon>0$ the running time of the $(1-\varepsilon)$-approximation algorithm of Alon and Sudakov is linear in the input length. 
%
%

\probAtMostClust can be seen as  a discrete variant of the well-known \textsc{$k$-Means Clustering}.   
This problem has been studied thoroughly, particularly in the areas of computational geometry and machine learning.  We refer to \cite{DBLP:journals/jacm/AgarwalHV04,DBLP:conf/stoc/BadoiuHI02,KumarSS10}  for further references to the works  on \textsc{$k$-Means Clustering}. In particular, the ideas from the algorithm  for \textsc{$k$-Means Clustering} of Kumar et al. 
\cite{KumarSS10}  form  the starting point of our algorithm for \rclustering.

 \subsection{Our approach.}

 \paragraph{Sampling lemma and deterministic algorithm} 

Our algorithms are based on Sampling Lemma (Lemma~\ref{lem:samplingMain}). 
Suppose we have a relation $R\subseteq \{0,1\}^k$ and weight tuple $\bfw=(w_1,\ldots,w_k)$, where $w_i\geq 0$ for all $i\in \{1,\ldots,k\}$. 
Then Sampling Lemma says that for any $\epsilon>0$, there is a constant  $r=\Theta(\frac{k}{\varepsilon^2} \log \frac{1}{\varepsilon})$ such that for any tuple $\bfp=(p_1,\ldots,p_k)$, $0\leq p_i\leq 1$, $r$ random samples from Bernoulli distribution $B(p_i)$ for each $i\in \{1,\ldots,k\}$ gives a {\em good estimate} of the minimum weighted distance of $\bfp$ from the tuples in $R$.  For more details we refer 
to  Lemma~\ref{lem:samplingMain}.

Here we explain how our sampling lemma works to design a PTAS.
 Let  $J=(X,k,{\cal R}=\{R_1,\ldots,R_{m}\})$ of \rclustering, $\varepsilon>0$ be an instance \rclustering\ and let $C=\{\bfc_1,\ldots,\bfc_k\}$ be an optimum solution to $J$. Let $X_1\uplus X_2\ldots \uplus X_k $ be a partition of $X$ such that $\sum_{\bfx\in X} \hdist(\bfx,C)=\sum_{i=1}^k \sum_{\bfx\in X_i} \hdist(\bfx,\bfc_i)$. Informally, using Sampling Lemma, we prove that there is a constant  $r=\Theta(\frac{k}{\varepsilon^2} \log \frac{1}{\varepsilon})$ such that  given $w_i=\vert X_i\vert$ and $r$ vectors uniformly at random chosen with repetition from $X_i$ for all $i\in \{1,\ldots,k\}$, 
 we can compute in linear time a $(1+\epsilon)$-approximate solution to $J$ (see proof of Lemma~\ref{lem:samplelargecluster0}). This immediately implies a PTAS for the problem.   

 \paragraph{Linear time algorithm (Theorem~\ref{thm:mainrclustering})} 
 
 
The general idea of
our algorithm for \rclustering\ is inspired by the  algorithm of Kumar et al.~\cite{KumarSS10}.
  Very informally,  the  algorithm of Kumar et al.~\cite{KumarSS10} for  \kmclust  is based on repeated sampling and does the following. 
For any (optimum) solution there is a cluster which size is at least $\frac{1}{k}$-th of the  number of input vectors. 
Then when we sample a constant number of 
vectors from the input uniformly at random, with a good probability, the sampled vectors  will be from the largest cluster. 
Moreover, if we sample sufficiently many (but still constant) number of vectors, they not only will belong to the largest cluster with a good probability, but  taking  the mean  of the sample as the center of the whole cluster in the solution, we obtain a vector ``close to optimum''.
This procedure  succeeds if the size of the largest cluster is   a large fraction of the number of vectors we used to sample from. 
Then the main idea behind the algorithm of Kumar et al. to assign vectors at a small distance from the guessed center vectors to their  clusters.  Moreover, once some vectors are assigned to clusters, then the next largest cluster  will be a constant (depending on $k$ and $\epsilon$) fraction the size of  the yet unassigned  vectors. 
With the right choice of parameters, it is possible to show that with a good probability this procedure will be a good approximation to the optimum solution.

On a very superficial level we want to implement a similar procedure: iteratively  sample, identify centers from samples, assign some of the unassigned vectors to the centers, then  sample again, identify centers,  etc.  Unfortunately, it is not that simple. The main difficulty  is that   in  
  \rclustering, even though we could   guess vectors from  the largest cluster, we cannot select a center vector    for this cluster because  the centers of ``future'' clusters
  should satisfy constrains from $\cR$---selection of one center could influence the ``future'' in a very bad way. Since we cannot select a good center, we cannot assign vectors to the cluster, and thus we cannot guarantee that sampling will   find the next cluster. The whole procedure just falls apart!
  
  Surprisingly, the sampling idea still works for   \rclustering but we have to be more careful. The main idea behind our approach is that if we sample all ``big'' clusters simultaneously  and assign the centers to these clusters such that the assigned centers ``partially'' satisfy $\cR$, then with a good probability this choice does not mess up much the solution.  After sampling vectors from all big clusters, we are left with two jobs-- (i) find centers for the clusters sampled simultaneously and $(ii)$ these centers will be a subset of our final solution.
  The condition $(i)$ is guaranteed by our new Sampling Lemma (Lemma~\ref{lem:samplelargecluster}). Towards maintaining condition $(ii)$, we prove that 
  even after finding ``approximately close centers'' for the big clusters, there exist centers for the small clusters which together with 
  the already found centers is a {\em good} approximate solution (i.e, they obey the relations and its cost is small).  
    As far as we succeed in finding with a good probability a subset of  ``good'' center vectors, we assign  some of the remaining input vectors to the clusters around the centers. Then we can proceed iteratively.   
  
Now we explain briefly how to obtain the running time to be linear. In each iteration 
after finding some center vectors we have mentioned that the vectors in the remaining input vectors which are close to already found centers can be {\em safely} assigned to the clusters of the center vectors already found. In fact we show that if the number of such vectors (vectors which are close to already found centers) are at most half the fraction of remaining input vectors, then there exist at least one cluster (whose center is yet be computed) which contains a constant fraction of the remaining set of vectors. In the other case we have that half fraction of the remaining vectors can be assigned to already found centers. This leads to a recurrence relation $T(n,k)=T(\frac{n}{2},k)+cT(n,k-k')+c'n\cdot m$ and $T(n,0)=T(0,k)=1$, where $c$ and $c'$ are constants depending on $k$ and $\epsilon$, and $k'\geq 1$, provided we could find approximate cluster centers from the samples of large clusters 
in linear time.  The above recurrence will solves to $f(k,\epsilon)n\cdot m$, for some function $f$.  
We need to make sure that we can compute cluster centers from the samples. In the case of designing a PTAS, we have already explained that we could compute approximate cluster centers using samples {\em if know the size of each of those clusters}.  In fact we show that if the sizes of large clusters are {\em comparable} and know them {\em approximately}, then we could compute approximate cluster centers in linear time (see Lemma~\ref{lem:samplelargecluster}).

In Section~\ref{sec:prelims}, we give notations, definitions and some known results which we use through out the paper. In Section~\ref{sec:notations} we give notations related to \rclustering. In Section~\ref{sec:samplinglemma}, we prove the important Sampling Lemma, which we use to design both PTAS and linear time randomized approximation scheme for \rclustering.  Then in Section~\ref{sec:ptas}, we prove how Sampling Lemma can be used to get a (deterministic) PTAS for the problem. The subsequent sections are building towards obtaining a linear time approximation scheme for the problem.

\section{Preliminaries}
\label{sec:prelims}

We use ${\mathbb N}$ to denote the set $\{1,2,\ldots\}$. 
For an integer $n\in {\mathbb N}$, we use $[n]$ as a shorthand for $\{1,\ldots,n\}$. 
For a set $U$ and non-negative integer $i$, $2^U$ and $\binom{U}{i}$ denote 
the set of subsets of $U$ and set of $i$ sized subsets of $U$, respectively.  
For a tuple $b=(b_1,\ldots,b_k)\in \{0,1\}^k$ and an index $i\in \{1,\dots, k\}$, we use  $b[i]$ denotes 
the $i$th entry of $t$, i.e, $b[i]=b_i$. 
We use $\log$ to denote the logarithm with base $2$.   

 In the course of our algorithm, we will be construction a solution  iteratively. When we find a set of vectors $C=\{\bfc_1, \dots, \bfc_r\}$, $r<k$, which will be a part of the solution, these vectors should satisfy relations $\RR$. Thus we have to guarantee that for some index set $I\subset\{1,\dots, k\}$ of size $r$, the set of vectors $C$ satisfies the part of $\RR$ ``projected'' on $I$. 
 More precisely, 
\begin{definition}[Projection of $\RR$ on $I$,  $\proj{\RR}{I}$]
Let $R\subseteq \{0,1\}^k$ be a relation and $I=\{i_1,\ldots,i_r\} \subseteq \{1,\dots, k\}$ be a subset of indices, 
where $i_1<i_2<\cdots<i_r$.  We say that a relation  $R'\subseteq \{0,1\}^r$ is a \emph{projection  of $R$ on  $I$}, denoted by $\proj{R}{I}$,  if 
  $R'$  is a  set of $r$-tuples   from 
  $ \{0,1\}^{r}$  such that 
$u=(u_{1},\ldots,u_{r})\in \proj{R}{I}$ if and only if there exists $t\in R$ such that $t[i_j]=u_{j}$ for 
all $j\in \{1,\dots, r\}$. In other words, the tuples of  $\proj{R}{I}$ are obtained from tuples of $R$ by leaving only the entries with coordinates from $I$.
 For a family $\RR=\{R_1,\ldots,R_m\}$ of relations, where $R_i\subseteq \{0,1\}^k$,  
we use $\proj{\RR}{I}$ to denote the family 
$\{\proj{R_1}{I},\ldots,\proj{R_m}{I}\}$.
\end{definition}

Thus a set of vectors $\bfc_1, \dots, \bfc_r\in  \{0,1\}^m$ satisfies $\proj{\RR}{I}$ if and only if for every $\ell\in \{1,\dots, m\}$ there exists $\bft\in R_\ell$ such for every  $j\in \{1,\dots, r\}$,    $\bfc_\ell[j]=\bft[i_j]$. As far as we fix a part of the solution $C=\{\bfc_1, \dots, \bfc_r\}$ and index set $I$ of size $r$, such that $C$ satisfies   $\proj{\RR}{I}$,
we can reduce the family of relations $\RR$ by deleting from each relation $R_i\in \RR$ all $k$-tuples not compatible with $C$ and $I$. 
More precisely, for every $1\leq i \leq m$,  we can leave only  $k$-tuples which projections on $I$ are equal to
$   \{\bfc_1[i], \dots, \bfc_r[i]\}$.  Let the reduced family of relations be  $\rest{\RR}{I}{C}$.
Then in every 
solution $S$ extending $C$, the set of  vectors $S\setminus C$ should satisfy  the projection of $\rest{\RR}{I}{C}$ on $\bar{I}=\{1,\dots, k\} \setminus I$. This brings us to the following definitions. 

\begin{definition}[Reducing relations $\RR$ to  $\rest{\RR}{I}{C}$]
Let $R\subseteq \{0,1\}^k$ be a relation and $I=\{i_1,\ldots,i_r\} \subseteq \{1,\dots, k\}$ be a subset of indices, 
where $i_1<i_2<\cdots<i_r$,  and let $u=(u_{1},\ldots,u_{r})\in \proj{R}{I}$ be an 
  $r$-tuple.  We say that relation $R'\subseteq R$ is \emph{obtained from  $R$ subject to $I$ and $u$} and write $R'=\rest{R}{I}{u}$, if 
  \[
  R'=\{t \in R~ |~ t[i_j]=u_{j}  \mbox{ for all } j\in \{1,\dots, r\}\}.
  \]
  For a set of vectors 
  $C=\{\bfc_1, \dots, \bfc_r\}$, set  $I\subseteq \{1,\dots, k\}$ of size $r$, and a family of relations 
   $\RR=\{R_1, \dots, R_m\}$,
  we denote by $\rest{\RR}{I}{C}$ the family of relations 
  $\{R_1', \dots, R_m'\}$, where $R_i'= \rest{R_i}{I}{(\bfc_1[i],\ldots,\bfc_{r}[i])}$, $1\leq i\leq m$.
  
   \end{definition}

  
%
\begin{definition}[$\RR(I,C)$: Projection of  $\rest{\RR}{I}{C}$ on  $\bar{I}$ ]
For relation  $R\subseteq \{0,1\}^k$,  $r$-sized subset of indices $I\subseteq \{1,\dots, k\}$  and $r$-tuple  $u=(u_{1},\ldots,u_{r})\in \proj{R}{I}$, we  use  $R(I,u)$  to denote the projection of $\rest{R}{I}{u}$ on  $\bar{I}=\{1,\dots, k\} \setminus I$.

For a    family  $\RR=\{R_1,\ldots,R_m\}$ of relations, set of $r\leq k$ vectors $C=\{\bfc_1,\ldots,\bfc_r\}$ from $\{0,1\}^m$, and $r$-sized set of indices $I\subseteq \{1,\dots, k\}$, 
we  use $\RR(I,C)$ to denote the family $\{R_1',\ldots,R_m'\}$, where 
 $R_i'=R_i(I,t_i)=
\proj{\rest{R_{i}}{I}{t_i}}{\overline{I}}$ and $t_i=(\bfc_1[i],\ldots,\bfc_{r}[i])$.
 \end{definition}

 In other words, $R(I,u)$ consists of all $(k-r)$-tuples $v$, such that  ``merging'' of $u$ and $v$ results in a $k$-tuple from $R$.
  In particular, the extension of $u$ and $I$ in $R$ can be generated by ``merging'' $u$ and all vectors of   $R(I,u)$.

We also use $\mathbf{0}$ and $\mathbf{1}$ to denote   vectors with all entries equal to $0$ and $1$, respectively, where the dimension of 
the vectors will be clear from the context. 
For vector $\bfx \in \{0,1\}^m$ and set $X\subset \{0,1\}^m$, we use $\hdist(\bfx,C)$ to denote the minimum Hamming distance (the number of different coordinates) 
between $\bfx$ and  vectors in $C$. 
 For sets $X,Y\subset \{0,1\}^m$, we 
define  \[\Hdist(X,Y)=\sum_{\bfx\in X}\hdist(\bfx,Y).\]  
For a vector $\bfx\in \{0,1\}^m$ and  integer $\ell>0$, we use $\BB(\bfx,\ell)$ to denote the open ball  of radius $\ell$ centered in $x$, that is, 
the set of vectors   in  $  \{0,1\}^m$ at Hamming distance    less than $\ell$ from $\bfx$.


%


\paragraph*{Probability.} 
In the analysis of our algorithm we will be using well known tail inequalities like Markov's  and Hoeffding's inequalities. 
%
\begin{proposition}[Markov's inequality~\cite{Mitzenmacherbook}]
\label{prop:markov}
Let $X$ be a non-negative random variable and $a > 0$. Then 
\[
\Pr (X\geq a\cdot \expect[X])\leq \frac {1}{a}.
\]
\end{proposition}

\begin{proposition}[Hoeffding's inequality~\cite{Hoeffding63}]
\label{prop:Hoeffding}
Let $X_1,\ldots, X_n$ be independent random variables such that each $X_i$ is strictly bounded by the intervals $[a_i, b_i]$. 
Let $X=\sum_1^{n}X_i$ and $t>0$. Then   
\[
\Pr(X-\expect[X]\geq t)\leq e^{\left(-\frac{2t^2}{\sum_{i\in [n]} (b_i-a_i)^2}\right)}.
\]
\end{proposition}

\section{Notations related to \rclustering}\label{sec:notations}

Let $J=(X,k,{\cal R}=\{R_1,\ldots,R_m\})$ be an instance of \rclustering\  
and $C=\{\bfc_1,\ldots,\bfc_k\}$ be a solution to $J$, that is, a set of vectors satisfying ${\cal R}$. Then the cost of  
$C$ is $\Hdist(X,C)$. Given set $C$, there is a natural way we can partition the set of vectors $X$ into $k$ 
sets $X_1\uplus\cdots\uplus X_k$ such that \[\Hdist(X,C)=\sum_{i=1}^{k}\Hdist(X_i,\{\bfc_i\})=\sum_{i=1}^{k}\sum_{\bfx\in X_i}\hdist(\bfx,\bfc_i).\] 
Thus for each vector $\bfx$ in $X_i$, the closest to $\bfx$ vector from $C$ is  $\bfc_i$. 
We call  such partition \emph{clustering of $X$ induced by $C$} and refer to sets $X_1,\ldots,X_k$ as to  {\em clusters corresponding to   $C$}. 

We use  $\opt(J)$  to denote the optimal solution to $J$. That is
\[
\opt(J)=\min \{  \Hdist(X,C) ~| ~ <C,\cR>  \}.
\]
Note that in the definition of a vector set $C$ satisfiying relations $\RR$, we require that the size of $C$ is $k$. We also need a relaxed notion for vector sets of size smaller than $k$ to satisfy a part of $\RR$. 
  
 \begin{definition}[Vectors respecting $\cR$] 
  Let  $C=\{\bfc_1,\ldots,\bfc_{i}\}\subseteq \{0,1\}^m$ be a set of binary vectors, where $i\leq k$, we say that $C$ respects 
  ${\cal R}$ if 
there is an index set 
$I\in \{1,\dots, k\}$ such that $<C,\proj{\RR}{I}>$, that is,  $C$ satisfies $\proj{\RR}{I}$. 
In other words, $C $ is a solution to $(X,i,\proj{\RR}{I})$. 
\end{definition}

Notice that given a set $C'$ of 
$i \leq k$ vectors  which respects $\RR$, one can extend it to a set $C$   in time linear in the size of $J$ such that $C$ satisfies $\cal R$. Thus $C$ is a (maybe non-optimal) solution to $J$ such that $\Hdist(X,C)\leq \Hdist(X,C')$. We will use this observation in several places and thus state it as a proposition. 

\begin{proposition}
\label{prop:smallsolution}
Let $J=(X,k,{\RR}=\{R_1,\ldots,R_m\})$ be an instance of \rclustering\ and 
$C'=\{\bfc_1,\cdots, \bfc_i\}\subseteq \{0,1\}^{m}$ for  $i\leq k$ be a set of vectors  respecting $\RR$. 
Then 
 there is linear time algorithm 
which finds a  solution $C$ to $J$ such that $\Hdist(X,C)\leq \Hdist(X,C')$.  
\end{proposition}

\begin{definition}
Let $J=(X,k,{\cal R})$ be an instance of \rclustering. For $i\in \{1,\dots, k\}$, we define
\[
\opt_{i}(J)=\min\{ \Hdist(X,C\}~ |~ |C|=i \text{ and } C \text{ respects } {\cal R} \}.
\]  
\end{definition}
An equivalent way of defining $\opt_i(J)$ is 
\[
\opt_{i}(J)=\min\{ \opt(X,i,\proj{\RR}{I})~ |~ I\subseteq \{1,\dots, k\} \text{ and } |I|=i\}.
 \]  
%
Notice that \[\opt_1(J)\geq  \opt_2(J) \geq \cdots \geq \opt_k(J)=\opt(J).\]

\section{Sampling probability distributions}\label{sec:samplinglemma}


One of the main ingredient of our algorithms is the lemma about sampling of specific probability distributions. 
To state the lemma we use the following notations. For a real $p$ between $0$ and $1$ we will denote by $B(p)$ the Bernoulli distribution which assigns probability $p$ to $1$ and $1-p$ to $0$. 
We will write $X \sim B(p)$ to denote that $X$ is a random variable with distribution $B(p)$.
\begin{definition}[Weighted distance $d^{{\bfw}}$]
For two $k$-tuples $\bfu=(u_1,\ldots,u_k)$ and  $\bfv=(v_1,\ldots,v_k) $ over reals and $k$-tuple  $\bfw=(w_1,\ldots,w_k)$ with $w_i\geq 0$, 
the {\em distance from $\bfu$ to $\bfv$ weighted by $\bfw$} is defined as  
$$d^{{\bfw}}({\bfu}, {\bfv}) = \sum_{i=1}^{k} w_i|u_i-v_i|.$$

\end{definition}


Informally, Sampling Lemma proves the following. For an integer $k$,  a relation $R\subseteq \{0,1\}^k$, a sequence of probability 
distribution $\bfp=(p_1,\ldots,p_k)$ and $\epsilon>0$, there is a constant $r$ (depending on $k$ and $\epsilon$) such that for every $1 \leq i \leq k$ a sample of $r$ random values from $B(p_i)$ gives us a tuple $\bfrho\in R$ which is a {\em good} estimate of  $d^{ \bfw}( \bfrho,  \bfp)$. 

\begin{lemma}[Sampling Lemma]
\label{lem:samplingMain}
There exists   $c > 0$ such that for every 
 $\epsilon > 0$,
positive integers $k$ and 
 $r \geq c \cdot \frac{k}{\epsilon^2}\cdot \log \frac{1}{\epsilon}$,
$k$-tuples $ {\bfp} = (p_1, \ldots, p_k)$ with $0 \leq p_i \leq 1$,
and  $  \bfw = (w_1, \ldots, w_k)$ with $0 \leq w_i$, and relation 
 $R \subseteq \{0,1\}^k$, the following is satisfied.

For every $1 \leq i \leq k$ and $1 \leq j \leq r$, let  $X_i^j \sim B(p_i)$, and let $ {\bfQ} = (Q_1, \ldots Q_k)$ be the $k$-tuple  of random variables,  where $Q_i = \frac{1}{r}\sum_{j=1}^r X_i^j$.
Let $d_{min} $ be the minimum distance weighted by $\bfw$ from $\bfp$ to a $k$-tuple from $R$.    
  Let $ {\bfrho}$ be a $k$-tuple from  $R$ within the minimum weighted by $\bfw$ distance to $\bfQ$, 
  that is, $\bfrho=\argmin_{\bfx\in R}d^\bfw(\bfx,\bfQ)$, 
 and let  
 $D = d^{ \bfw}( \bfrho,  \bfp)$. Then
$\expect[D] \leq (1+\epsilon)d_{min}.$
\end{lemma}

\begin{proof}
Let  $\bfu=\argmin_{\bfx\in R}d^\bfw(\bfx,\bfp)$. Then $d_{min} = d^{ \bfw}( \bfu,  \bfp)$.
%
 Let $R_{small}$ be the set of all tuples $\bfv \in R$ such that  
$d^{\bfw}(\bfv, \bfp) \leq (1+\frac{\epsilon}{2})d_{min}$. Let $R_{big} = R \setminus R_{small}$. 
%
%
We will prove the following claim.

\begin{claim}\label{clm:farGivesLowProb}
For every $\bfv \in R_{big}$, 
$$\Pr(d^{\bfw}(\bfv, \bfQ) \leq d^{\bfw}(\bfu, \bfQ)) \leq \frac{d_{min}}{d^{\bfw}(\bfv, \bfp)} \cdot \frac{\epsilon}{2^{k+1}}.$$
\end{claim}

Assuming Claim~\ref{clm:farGivesLowProb} we complete the proof of the lemma:

\begin{align*}
\expect[D] = \sum_{v \in R_{small}} d^{\bfw}(\bfv, \bfp) \cdot \Pr(\bfrho = \bfv) + \sum_{v \in R_{big}} d^{\bfw}(\bfv, \bfp) \cdot \Pr(\bfrho = \bfv) \\
\leq d_{min}(1+\frac{\epsilon}{2}) + \sum_{v \in R_{big}} d^{\bfw}(\bfv, \bfp) \cdot \Pr(d^{\bfw}(\bfv, \bfQ) \leq d^{\bfw}(\bfu, \bfQ)) \\
\leq d_{min}(1+\frac{\epsilon}{2}) + \sum_{v \in R_{big}} d^{\bfw}(\bfv, \bfp) \cdot \frac{d_{min}}{d^{\bfw}(\bfv, \bfp)} \cdot \frac{\epsilon}{2^{k+1}} \\
\leq d_{min}(1+\frac{\epsilon}{2}) + d_{min} \cdot \frac{\epsilon}{2} \leq d_{min}(1+\epsilon). 
\end{align*}

Hence, all that remains to prove the lemma is to prove Claim~\ref{clm:farGivesLowProb}.

\begin{proof}[Proof of Claim~\ref{clm:farGivesLowProb}]
We will assume without loss of generality that $\epsilon \leq \frac{1}{10}$. By renaming $0$ to $1$ and vice versa at the coordinates $i$ where $u_i = 1$,  we may assume that $\bfu = \mathbf{0}$. Thus $d_{min} = d^{\bfw}(\mathbf{0}, \bfp)$. We may now rewrite the statement of the claim as:

\begin{align}\label{eqn:tryToProve}
\Pr(d^{\bfw}(\bfv, \bfQ) \leq d^{\bfw}(\mathbf{0}, \bfQ)) \leq \frac{d^{\bfw}(\mathbf{0}, \bfp)}{d^{\bfw}(\bfv, \bfp)} \cdot \frac{\epsilon}{2^{k+1}}.
\end{align}

Consider now the weight $k$-tuple $\bfw' = (w_1', \ldots, w_k')$ where $w_i' = w_i$ if $v_i = 1$ and $w_i' = 0$ if $v_i = 0$. We have that $\Pr(d^{\bfw}(\bfv, \bfQ) \leq d^{\bfw}(\mathbf{0}, \bfQ)) = \Pr(d^{\bfw'}(\bfv, \bfQ) \leq d^{\bfw'}(\mathbf{0}, \bfQ))$, and that $\frac{d^{\bfw'}(\mathbf{0}, \bfp)}{d^{\bfw'}(\bfv, \bfp)} \leq \frac{d^{\bfw}(\mathbf{0}, \bfp)}{d^{\bfw}(\bfv, \bfp)}$. Hence, in order to prove   \eqref{eqn:tryToProve},  it is sufficient to prove

\begin{align*}
\Pr(d^{\bfw'}(\bfv, \bfQ) \leq d^{\bfw'}(\mathbf{0}, \bfQ)) \leq \frac{d^{\bfw'}(\mathbf{0}, \bfp)}{d^{\bfw'}(\bfv, \bfp)} \cdot \frac{\epsilon}{2^{k+1}}.
\end{align*}

In other words,  it is sufficient to prove  \eqref{eqn:tryToProve} under the additional assumption that $w_i = 0$ whenever $v_i = 0$. Under this assumption we have that $d^{\bfw}(\bfv, \bfQ) = d^{\bfw}(\mathbf{1}, \bfQ)$, and that $d^{\bfw}(\bfv, \bfp) = d^{\bfw}(\mathbf{1}, \bfp)$. Thus, it is suffices to prove that $d^{\bfw}(\mathbf{1}, \bfp) \geq (1 + \frac{\epsilon}{2}) \cdot d^{\bfw}(\mathbf{0}, \bfp)$ implies that

\begin{align}\label{eqn:newToProve}
\Pr(d^{\bfw}(\mathbf{1}, \bfQ) \leq d^{\bfw}(\mathbf{0}, \bfQ)) \leq \frac{d^{\bfw}(\mathbf{0}, \bfp)}{d^{\bfw}(\mathbf{1}, \bfp)} \cdot \frac{\epsilon}{2^{k+1}}.
\end{align}

Let $w^* = \sum_{i = 1}^k w_i$. We have that $d^{\bfw}(\mathbf{0}, \bfp) + d^{\bfw}(\mathbf{1}, \bfp) = w^*$.
Thus $d^{\bfw}(\mathbf{1}, \bfp) \geq (1 + \frac{\epsilon}{2}) \cdot d^{\bfw}(\mathbf{0}, \bfp)$ implies that $d^{\bfw}(\mathbf{0}, \bfp) \leq \frac{w^*}{2+\frac{\epsilon}{2}}$. 
We also have that $d^{\bfw}(\mathbf{1}, \bfQ) + d^{\bfw}(\mathbf{0}, \bfQ) = w^*$, and therefore $d^{\bfw}(\mathbf{1}, \bfQ) \leq d^{\bfw}(\mathbf{0}, \bfQ)$ if and only if $d^{\bfw}(\mathbf{0}, \bfQ) \geq \frac{w^*}{2}$. Furthermore, $d^{\bfw}(\mathbf{1}, \bfp) \leq w^*$.
Hence, to prove the claim (in partiular Equation~\ref{eqn:newToProve}) it is sufficient to show that $d^{\bfw}(\mathbf{0}, \bfp) \leq \frac{w^*}{2+\frac{\epsilon}{2}}$ implies
\begin{align}\label{eqn:finalToProve}
\Pr(d^{\bfw}(\mathbf{0}, \bfQ) \geq \frac{w^*}{2}) \leq \frac{d^{\bfw}(\mathbf{0}, \bfp)}{w^*} \cdot \frac{\epsilon}{2^{k+1}}.
\end{align}

We now prove Equation~\ref{eqn:finalToProve} distinguishing between two cases.

\smallskip\noindent
{\bf Case 1,} {\em $d^{\bfw}(\mathbf{0}, \bfp) > \frac{w^*}{128k}$}. \\
\smallskip
\noindent
We have that $d^{\bfw}(\mathbf{0}, \bfQ) = \sum_{i=1}^k\sum_{j=1}^r \frac{w_i}{r}X_i^j$.
Thus, $d^{\bfw}(\mathbf{0}, \bfQ)$ is the sum of $kr$ independent random variables, grouped into groups of size $r$, where all variables in group $i$ take value $\frac{w_i}{r}$ with probability $p_i$ and value $0$ with probability $1-p_i$.
It follows that $\expect[d^{\bfw}(\mathbf{0}, \bfQ)] = d^{\bfw}(\mathbf{0}, \bfp) \leq \frac{w^*}{2+\frac{\epsilon}{2}}$. Hence $\frac{w^*}{2} - \expect[d^{\bfw}(\mathbf{0}, \bfQ)] \geq \frac{w^*}{2}- \frac{w^*}{2+\frac{\epsilon}{2}} > \frac{\epsilon w^*}{10}$, where the last inequality follows from the assumption that $\epsilon \leq \frac{1}{10}$. Thus we may use Proposition~\ref{prop:Hoeffding} to upper bound $\Pr(d^{\bfw}(\mathbf{0}, \bfQ) \geq \frac{w^*}{2})$.

\begin{align*}
\Pr(d^{\bfw}(\mathbf{0}, \bfQ) \geq \frac{w^*}{2}) \leq \Pr(d^{\bfw}(\mathbf{0}, \bfQ) - \expect[d^{\bfw}(\mathbf{0}, \bfQ)] > \frac{\epsilon w^*}{10})
\leq \exp\left(-\frac{2\epsilon^2(w^*)^2}{100\Sigma_{i=1}^k \Sigma_{j=1}^r \left(\frac{w_i}{r}\right)^2}\right) \\
\leq \exp\left( -\frac{\epsilon^2r}{50}\right) 
\leq \frac{w^*}{128k} \cdot \frac{128k}{w^*} \cdot \frac{\epsilon}{128k \cdot 2^{k+1}} 
\leq \frac{d^{\bfw}(\mathbf{0}, \bfp)}{w^*} \cdot \frac{\epsilon}{2^{k+1}}. 
\end{align*}
Here the second transition is by Proposition~\ref{prop:Hoeffding}, while the fourth is by the choice of $r = \Omega(\frac{k}{\epsilon^2}\cdot \log \frac{1}{\epsilon})$.

\medskip\noindent
{\bf Case 2,} {\em $d^{\bfw}(\mathbf{0}, \bfp) \leq \frac{w^*}{128k}$}. \\
\noindent
For every $i$ such that $w_i \geq \frac{w^*}{4k}$,  we have $p_i \leq \frac{1}{32}$. Since $Q_i$ is binomially distributed we have that



\begin{align*}
\Pr(Q_i \geq \frac{1}{4}) = \sum_{t=\lceil \frac{r}{4}\rceil}^r {r \choose t} p_i^t(1-p_i)^{r-t} 
\leq 2^r \cdot p_i^{\frac{r}{4}} 
\leq p_i \cdot 2^r \cdot (\frac{1}{32})^{\frac{r}{4}-1} 
\leq p_i \cdot 2^{-\frac{r}{4}-5} 
\leq \frac{p_i}{4k^2} \cdot \frac{\epsilon}{2^{k+1}}.
\end{align*}
Here the last inequality follows from the fact that of $r = \Omega(\frac{k}{\epsilon^2}\cdot \log \frac{1}{\epsilon})$

For every $i$ such that $w_i \geq \frac{w^*}{4k}$ we have that $p_i \leq \frac{d^{\bfw}(\mathbf{0}, \bfp)}{w^*} \cdot 4k$, since otherwise
$w_ip_i > d^{\bfw}(\mathbf{0}, \bfp)$, a contradiction. Thus, for every such $i$ we have 
$$\Pr(Q_i \geq \frac{1}{4}) \leq \frac{d^{\bfw}(\mathbf{0}, \bfp)}{w^* k} \cdot \frac{\epsilon}{2^{k+1}}.$$
By the union bound,  we have that
\[
\Pr\left(\sum_{i ~:~ w_i \geq \frac{w^*}{4k}} w_iQ_i \geq \frac{w^*}{4}\right) \leq   \frac{d^{\bfw}(\mathbf{0}, \bfp)}{w^*} \cdot \frac{\epsilon}{2^{k+1}}.
\]
Since 
\[
\sum_{i : w_i < \frac{w^*}{4k}} w_iQ_i < \frac{w^*}{4}
\]
(with probability $1$), it follows that 
$$\Pr( d^{\bfw}(\mathbf{0}, \bfQ) \geq \frac{w^*}{2} ) \leq   \frac{d^{\bfw}(\mathbf{0}, \bfp)}{w^*} \cdot \frac{\epsilon}{2^{k+1}}.$$
This proves the claim, and completes the proof of Lemma~\ref{lem:samplingMain}. 
\end{proof}\end{proof}

\section{Warm up: Deterministic PTAS}\label{sec:ptas}
As a warm up, let us show how Sampling Lemma can be used to obtain a deterministic PTAS for  \rclustering. 
%
Towards that we need the following definition. 

\begin{definition}
\label{def:fidef}
Let $k,m\in {\mathbb N}$ and ${\cal R}=\{R_1,\ldots,R_{m}\}$ be a family of relations, where $R_i\subseteq \{0,1\}^k$ for each $i \in \{1,\ldots,m\}$. Let $\{S_1,\ldots, S_k\}$  be a family of multisets of vectors from $\{0,1\}^m$ and $w_1,\ldots,w_k \in {\mathbb R}_{\geq 0}$. 
For a vector (multi)set $B\subseteq \{0,1\}^m$, let $z^{(i)}(B)$  be the number of vectors in $B$ with the $i$-{th}  entry 
equal to $0$ and let $d^{(i)}(B) =\vert B\vert-z^{(i)}(B)$ be the number of vectors in $B$ with the $i$-{th}  entry 
equal to $1$. Then $\best{{\cal R}}(S_1,\ldots,S_k,w_1,\ldots,w_k)$ is a set of vectors $\{\bfc_1,\ldots,\bfc_{k}\}$ satisfying ${\cal R}$  which is defined  as follows.   
For  $i\in \{1,\dots, m\} $ 
and a  $k$-tuple  $\bfb=(b_1,\ldots,b_{k})\in R_i$, let 
\begin{equation*}
\label{eqn:fi}
f_i(b_1,\ldots,b_{k})=\sum_{j\in I_b} w_j \cdot z^{(i)}(S_j) + \sum_{j\in [k]\setminus I_b} w_j \cdot d^{(i)}(S_j), 
\end{equation*}
where $I_b=\{j\in \{1,\dots, k\} \colon b_j=1 \}$. 
The set   $\{\bfc_1,\ldots,\bfc_{k}\}$ be  such that 
$(\bfc_1[i],\ldots,\bfc_{k}[i])\in R_i$ and  $f_i(\bfc_1[i],\ldots,\bfc_{k}[i])=\min_{\bfb\in R_i} f_i(\bfb)$,  $1\leq i\leq m$.
\end{definition}

\begin{lemma}
\label{lem:samplelargecluster0}
Let $J=(X,k,{\cal R}=\{R_1,\ldots,R_{m}\})$ of \rclustering, $\varepsilon>0$ and 
$r=\Theta(\frac{k}{\varepsilon^2} \log \frac{1}{\varepsilon})$ is the constant defined in Lemma~\ref{lem:samplingMain}. 
Then there exist $w_1,\ldots,w_k\in {\mathbb N}$ and a family $\{S'_1,\ldots,S'_k\}$ of $r$ sized  multisets of vectors  from $X$,  such that 
$$\Hdist(X,\best{{\cal R}}(S'_1,\ldots,S'_k,w_1,\ldots,w_k))\leq  (1+\varepsilon)\opt(J).$$
\end{lemma}

\begin{proof}

Let $C^*=\{\bfc_1^*,\ldots,\bfc_{k}^*\}$ be an optimal solution  to $J$ with corresponding clusters $P_1,\ldots,P_{k}$.  That is $\opt(J)= \Hdist(X,C^*)= \sum_{i=1}^{k}\Hdist(P_i,\{c^*_i\})$. For each $i\in \{1,\ldots,k\}$, we set $w_i=\vert P_i\vert$.  
For each $i\in \{1,\ldots,k\}$, we define a multiset  $S_i$  of $r$ vectors, where each vector in $S_i$ is chosen 
uniformly at random with repetition from $P_i$.  To prove the lemma it is enough to prove that 

\begin{equation}
\label{eqn:expbestran}
\expect[\Hdist(X,\best{{\cal R}}(S_1,\ldots,S_k,w_1,\ldots,w_k))]\leq  (1+\varepsilon)\opt(J)
\end{equation}

Recall the definition of functions $f_i, 1\leq i\leq m$ (see Definition~\ref{def:fidef}). 
For  $i\in \{1,\dots, m\} $ 
and a  $k$-tuple  $\bfb=(b_1,\ldots,b_{k})\in R_i$, 
\begin{equation}
\label{eqn:fi000}
f_i(b_1,\ldots,b_{k})=\sum_{j\in I_b} w_j \cdot z^{(i)}(S_j) + \sum_{j\in [k]\setminus I_b} w_j \cdot d^{(i)}(S_j), 
\end{equation}
where $I_b=\{j\in \{1,\dots, k\} \colon b_j=1 \}$. 
The set   $\best{{\cal R}}(S_1,\ldots,S_k,w_1,\ldots,w_k)=\{\bfc_1,\ldots,\bfc_{k}\}$ be  such that 
$(\bfc_1[i],\ldots,\bfc_{k}[i])\in R_i$ and  $f_i(\bfc_1[i],\ldots,\bfc_{k}[i])=\min_{\bfb\in R_i} f_i(\bfb)$,  $1\leq i\leq m$.

Notice that $\best{{\cal R}}(S_1,\ldots,S_k,w_1,\ldots,w_k)=\{\bfc_1,\ldots,\bfc_{k}\}$ is a random variable. 
We define a random variable 
$
Y=\sum_{i=1}^{k}\Hdist(P_i,\{\bfc_i\}).
$
To prove \eqref{eqn:expbestran}, it is enough prove that 
$\expect[Y]\leq (1+\varepsilon)\opt(J)$.
We define functions $g_i$ for all $i\in [m]$, which denotes the cost of  each element in the relation $R_i$ 
with respect to the partition $P_1\uplus\ldots \uplus P_{k}$ of $X$. Formally, 
for each $\bfb=(b_1,\ldots,b_{k})\in R_i$, we put
\begin{eqnarray}
g_i(\bfb)=\sum_{j\in I_b} z^{(i)}(P_j) + \sum_{j\in [k]\setminus I_b}d^{(i)}(P_j),  \label{eqngi000}
\end{eqnarray}
where $I_b=\{j\in \{1,\dots, k\} \colon b_j=1 \}$. 
Let $V_i=g_i(\bfc_1^*[i],\ldots,\bfc_{k}^*[i])$. Notice that $\opt(J)=\sum_{i=1}^{m}V_i$. 
Let $Y_i=g_i(\bfc_1[i],\ldots,\bfc_{k}[i])$, $1\leq i \leq m$. By \eqref{eqngi000} and the definition of $Y$, we have that $Y=\sum_{i=1}^{m}Y_i$. 
By the linearity of expectation, we have that  $\expect[Y]=\sum_{i=1}^{m}\expect[Y_i]$. 
Thus, to prove that $\expect[Y]\leq (1+\varepsilon)\opt(J)$, it is sufficient  to prove that 
  $\expect[Y_i]\leq (1+\varepsilon)V_i $ for every $i\in [m]$. 
\begin{claim}
\label{claim:yiexp}
For every  $i\in [m]$, $\expect[Y_i]\leq (1+\varepsilon)V_i$.
\end{claim}

\begin{proof}
Fix an index $i\in [m]$.  Let $z_j=z^{(i)}(P_j)$ and $d_j=d^{(i)}(P_j)$. Thus $z_j$ ($d_j$) is the number of vectors from $P_j$ whose $i$-th coordinate is $0$  ($1$). Let $n_j=\vert P_j\vert$ and $p_j=\frac{d_j}{n_j}$,   $j\in \{1,\dots, k\}$. 
Since each vector from 
$S_j$ is equally likely to be any vector in $P_j$,  
for every vector $v\in S_j$, 
\begin{equation}
\label{eqn:pi1000}
 \Pr(v[i]=1)
 =\frac{d_j}{n_j}=p_j \qquad\mbox{ and } \qquad  \Pr(v[i]=0)
 =\frac{z_j}{n_j}=1-p_j.  
\end{equation} 
Let $\bfp=(p_1,\ldots,p_{k})$.  We define  $k$-tuple   
$\bfw=(w_1,\ldots,w_k)$.  
and  claim that for every $\bfb=(b_1,\ldots,b_{k})\in R_i$, $d^{\bfy}(\bfb,\bfp)=g_i(\bfb)$. Indeed,
\begin{eqnarray}
d^{\bfw}(\bfb,\bfp)&=&\sum_{j=1}^{k} w_j \vert b_j - p_j \vert \nonumber\\
&=&\sum_{j\in I_b} w_j (1 - p_j)+\sum_{j\in [k]\setminus I_b} w_j  \cdot p_j \nonumber\\
&=& \sum_{j\in I_b} z_j+\sum_{j\in [k]\setminus I_b} d_j  \nonumber\\
&=&g_i(\bfb). \label{eqn:wgi1000}
\end{eqnarray}

For each set $S_j=\{v_1, \dots, v_r\}$, $1\leq j\leq k$   and $1\leq q\leq r$, we define random variable $X_{j}^{q}$ which is $1$ when $v_{q}[i]=1$ and $0$ otherwise.
By \eqref{eqn:pi1000}, we have that $X_{j}^{q}\sim B(p_j)$ for all $1 \leq j \leq k$ and $1 \leq q \leq r$. 
Let $\bfQ = (Q_1, \ldots Q_k)$ be the $k$-tuple of random variables,  where $Q_j = \frac{1}{r}\sum_{q=1}^r X_j^{q}$.
From the definitions of $z^{(i)}(S_j)$, $d^{(i)}(S_j)$ and $X_j^q$, we have that 
\begin{eqnarray}
z^{(i)}(S_j)=\sum_{q=1}^r (1-X_{j}^{q})  \quad \mbox{ and } \quad 
d^{(i)}(S_j)=\sum_{q=1}^r X_{j}^{q}. \label{eqn:zd000}
\end{eqnarray}

Let $\bfw=(w_1,\ldots,w_k)$. For   $\bfb\in R_i$, we 
  upper bound $d^\bfw(\bfb,\bfQ)$ in terms of $f_i(\bfb)$. 
\begin{eqnarray}
 d^\bfw(\bfb,\bfQ)&=&\sum_{j=1}^{k} w_j \vert b_i - Q_j \vert \nonumber\\
&\leq & \sum_{j\in I_b} w_j (1 - Q_j)+\sum_{j\in [k]\setminus I_b} w_j  \cdot Q_j\nonumber\\
&\leq &  \sum_{j\in I_b} w_j \cdot \left(1- \frac{1}{r}\sum_{q=1}^r X_{j}^{q}\right) + \sum_{j\in [k]\setminus I_b} w_j \cdot \left(\frac{1}{r}\sum_{q=1}^r X_{j}^{q}\right)\nonumber\\
&\leq & \frac{1}{ r} \left(\sum_{j\in I_b} w_j \cdot \left(\sum_{q=1}^r (1-X_{j}^{q})\right) + \sum_{j\in [k]\setminus I_b} w_j \cdot \left(\sum_{q=1}^r X_{j}^{q}\right)\right)\nonumber\\
&\leq & \frac{1}{r} \left( \sum_{j\in I_b} w_j \cdot z^{(i)}(S_j) + \sum_{j\in [k]\setminus I_b} w_j \cdot d^{(i)}(S_j)\right)\qquad\qquad\qquad \mbox{(by \eqref{eqn:zd000})}\nonumber\\
&\leq& \frac{1}{r} f_i(\bfb). \qquad\qquad\qquad\qquad (\mbox{by \eqref{eqn:fi000}}) \label{eqn:ficorrect}
\end{eqnarray}

Let  
\[
\bfrho=\argmin_{\bfx\in R_i} d^\bfw(\bfx,\bfQ)
.\]
By \eqref{eqn:ficorrect} and by the definition of the vector set $\{\bfc_1,\ldots,\bfc_{k}\}$, we have that $\bfrho=(\bfc_1[i],\ldots,\bfc_{k}[i])$. 
 
 We also define $k$-tuple
 \[
 \bfu=\argmin_{\bfx\in R_i} d^\bfw(\bfx,\bfp) 
.\]

This implies that
\begin{equation}
\label{eqn:dminapp2000}
d^\bfw(\bfu,\bfp) =g_i(\bfc_1^*[i],\ldots,\bfc_{k}^*[k])= V_i. 
\end{equation}

Thus the minimum weighted by $\bfw$ distance  $d_{min} $ from $\bfp$ to a $k$-tuple from $R_i$ is equal to $d^\bfw( \bfu,\bfp)$.    Let $D$ be the random variable which is a minimum weighted by $\bfw$ distance from $\bfrho$ to $\bfp$. By Lemma~\ref{lem:samplingMain}, 
\begin{equation}
\label{eq:samplingMain000}
\expect[d^\bfw(\bfrho,\bfp)]=  \expect[D] \leq (1+\epsilon)d_{min}=  (1+\epsilon) d^\bfw( \bfu,\bfp) . \end{equation}


Finally we upper bound $\expect[Y_i]$. 
\begin{eqnarray*}
\expect[Y_i]&=&\expect[g_i(\bfrho)]\\
&=&\expect[d^\bfw(\bfrho,\bfp)] \qquad\qquad\qquad\quad(\mbox{by \eqref{eqn:wgi1000}})\\
&\leq& (1+\varepsilon)\cdot d^\bfw( \bfu,\bfp) \qquad\qquad (\mbox{by  \eqref{eq:samplingMain000})}\nonumber\\
&\leq& (1+\varepsilon) V_i. \qquad\qquad\qquad (\mbox{by \eqref{eqn:dminapp2000}})
 \end{eqnarray*}
 This completes the proof of the claim.
\end{proof}

By Claim~\ref{claim:yiexp}, the fact that $\opt(J)=\sum_{i=1}^{m} V_i$ and by the linearity of expectation,  we have that 
\(
\expect[Y]\leq (1+\varepsilon) \opt(J) .\label{eqn:expYgivenE}
\)
This completes the proof of the lemma. 
\end{proof}

Lemma~\ref{lem:samplelargecluster0} implies Theorem~\ref{thm:mainrclusteringdet}. 
%
%
The remaining part of the paper is built towards obtaining a linear time randomized approximation scheme for \rclustering.


\section{Sampling instances with large clusters}\label{sec:samplinglemma1}

In this section we prove the algorithmic variant of  Sampling Lemma which will be the main engine of our randomized algorithm. Informally, the lemma says that   
if there is an optimal solution $C$ for the set of vectors $X$ such that  each of the clusters corresponding to $C$ contains a {large} fraction of the   vectors from $X$, then sampling { constantly} many vectors from $X$  for each cluster is a good estimate for a {good} approximate solution. 
 In fact, we need   a   stronger property:  We want to derive a good clustering of a subset of vectors $Z\subseteq X$ which is unknown to us (a hidden subset).

\begin{lemma}[Algorithmic Sampling Lemma]
\label{lem:samplelargecluster}
Let $X\subseteq \{0,1\}^m$ be a set of $n$  binary vectors and $Z \subseteq X$ be (an unknown) set of vectors. 
Let $J=(Z,k,{\cal R}=\{R_1,\ldots,R_{m}\})$ be an instance of \rclustering. 
Suppose that there exists a solution $C^*=\{\bfc_1^*,\ldots,\bfc_{k}^*\}$ (not necessarily optimal) to $J$ with corresponding clusters $P_1,\ldots,P_{k}$ and 
$1\geq \beta>0$ such that $\vert P_j\vert \geq n\beta$ for all $j\in \{1,\dots, k\}$.   We denote the cost of  $C^*$ by $V= \Hdist(Z,C^*)= \sum_{i=1}^{k}\Hdist(P_i,\{c^*_i\})$. 


\smallskip
Then there exists an algorithm $\A$ with the following specifications.  

\begin{itemize}\setlength\itemsep{-.8mm}
\item Input of $\A$ is  $X$, $k$, ${\RR}=\{R_1,\ldots,R_{m}\}$, $\delta,\varepsilon>0$, $0< \beta\leq 1$,
 and values 
$w_1,\ldots,w_{k}$ (promised bounds on the sizes of clusters $P_i$) such that for some constant $c$,  for each $j\in \{1,\dots, k\}$,
\[\frac{|P_j|}{c}\leq w_j\leq \frac{(1+\delta)|P_j|}{c}.\]  
\item  Output  of $\A$ is a  solution $C=\{c_1,\ldots,c_k\}$ to $J$ such that 
$\sum_{i=1}^k\Hdist(P_i,\{c_i\})\leq (1+\varepsilon)^2 (1+\delta)V$
with probability at least 
$\frac{\varepsilon \cdot \beta^{r\cdot k}}{1+\varepsilon}$, where   $r={\Theta\left (\frac{k}{\varepsilon^2} \log \frac{1}{\varepsilon}\right)}$,   and  
\item $\A$ runs in time $\OO\left( \left(\frac{k}{\varepsilon}\right)^2 \log \frac{1}{\varepsilon} \cdot  \sum_{i=1}^{m}\vert R_i \vert  \right)$.  
\end{itemize}



\end{lemma}

\begin{proof}
Let $r$ be the constant defined for $k $ and $\varepsilon $ in Lemma~\ref{lem:samplingMain}. 
That is  $r =\Theta\left(\frac{k}{\varepsilon^2}\cdot \log \frac{1}{\varepsilon}\right)$.
 For   vector $\bfb=(b_1,\ldots,b_{k})\in \{0,1\}^{k}$, we define $I_b=\{j\in \{1,\dots, k\} \colon b_j=1 \}$. 
For vector set $B\subseteq X$, let $z^{(i)}(B)$  be the number of vectors in $B$ with the $i$-{th}  entry 
equal to $0$. Similarly, let $d^{(i)}(B) =\vert B\vert-z^{(i)}(B)$ be the number of vectors in $B$ with the $i$-{th}  entry 
equal to $1$.

\paragraph*{Algorithm.}
The  algorithm $\A$ is very simple.  We  sample $k$ times (with possible repetitions)
uniformly at random  $r$ vectors from $X$.  Thus we obtain $k$ sets of vectors $S_1, \dots, S_k$, each of the sets is of size $r$. 
Based on these samples we 
 output solution $C=\{\bfc_1,\ldots,\bfc_{k}\}$  as follows.   
For  $i\in \{1,\dots, m\} $ 
and a  $k$-tuple  $\bfb=(b_1,\ldots,b_{k})\in R_i$, let 
\begin{equation}
\label{eqn:fi}
f_i(b_1,\ldots,b_{k})=\sum_{j\in I_b} w_j \cdot z^{(i)}(S_j) + \sum_{j\in [k]\setminus I_b} w_j \cdot d^{(i)}(S_j).
\end{equation}
Then $C$ is the set of vectors minimizing functions $f_i$ subject to constraints $\RR$. More precisely, we define  a vector set 
 $C=\{\bfc_1,\ldots,\bfc_{k}\}$    such that 
$(\bfc_1[i],\ldots,\bfc_{k}[i])\in R_i$ and  $f_i(\bfc_1[i],\ldots,\bfc_{k}[i])=\min_{\bfb\in R_i} f_i(\bfb)$,  $1\leq i\leq m$.
In other words, $C=\best{{\cal R}}(S_1,\ldots,S_k,w_1,\ldots,w_k)$. 
 Clearly, $C$ is a solution to $J$. 

\paragraph*{Running time.}
We assume that  input vectors $X$ are  stored in an array and that we can sample a vector u.a.r from a set of $n$ vectors stored in an 
array in constant time.  
For each $i\in [m]$ and $\bfb\in R_i$, the computation of $f_i(\bfb)$ takes time $\OO(r\cdot k)$. Then  computations of functions  $f_i(\bfb)$ for all 
$i\in [m]$ and $b\in R_i$  require  $\OO\left(\left(\sum_{i=1}^{m}\vert R_i \vert \right) \left(\frac{k}{\varepsilon}\right)^2 \log \frac{1}{\varepsilon}\right)$ time. For each $i$, we use an array of length $k$ to store the $k$-tuple from  $R_i$ which gives the minimum of $f_i$ computed 
so far during the computation. Therefore the running time of the algorithm follows.


\paragraph*{Correctness.}
Let ${\mathscr E}$ be the event that for all $j\in \{1,\dots, k\}$, $S_j \subseteq P_j $.  Since $\vert P_j\vert \geq \vert X\vert \cdot \beta$ for all $j\in \{1,\dots, k\}$, we have that 
\begin{equation}
\label{obs:prE}
 \Pr({\mathscr E}) \geq \beta^{r\cdot k}. 
\end{equation}

From  now on we assume that the event ${\mathscr E}$ happened.  
Therefore we can think that each vector in $S_j $ is chosen uniformly at random from $P_j$ (with repetitions). 
Notice that the output $C=\{\bfc_1,\ldots,\bfc_{k}\}$ is a random variable. 
We define a random variable 
$
Y=\sum_{i=1}^{k}\Hdist(P_j,\{\bfc_j\}).
$

Now we prove that 
$\expectation{Y}{\mathscr E}\leq (1+\varepsilon)(1+\delta)V$.
We define functions $g_i$ for all $i\in [m]$, which denotes the cost of  each element in the relation $R_i$ 
with respect to the partition $P_1\uplus\ldots \uplus P_{k}$ of $Z$. Formally, 
for each $\bfb=(b_1,\ldots,b_{k})\in R_i$, we put
\begin{eqnarray}
g_i(\bfb)=\sum_{j\in I_b} z^{(i)}(P_j) + \sum_{j\in [k]\setminus I_b}d^{(i)}(P_j).  \label{eqngi}
\end{eqnarray}
Let $V_i=g_i(\bfc_1^*[i],\ldots,\bfc_{k}^*[i])$. Notice that $V=\sum_{i=1}^{m}V_i$. 
Let $Y_i=g_i(\bfc_1[i],\ldots,\bfc_{k}[i])$, $1\leq i \leq m$. By \eqref{eqngi} and the definition of $Y$, we have that $Y=\sum_{i=1}^{m}Y_i$. 
By the linearity of conditional expectation, we have that  $\expectation{Y}{\mathscr E}=\sum_{i=1}^{m}\expectation{Y_i}{\mathscr E}$. 
Thus, to prove that $\expectation{Y}{\mathscr E}\leq (1+\varepsilon)(1+\delta)V$, it is sufficient  to prove that 
  $\expectation{Y_i}{\mathscr E}\leq (1+\varepsilon)(1+\delta)V_i $ for every $i\in [m]$. 
\begin{claim}
\label{claim:yiexp}
For every  $i\in [m]$, $\expectation{Y_i}{{\mathscr E}}\leq (1+\varepsilon)(1+\delta)V_i$.
\end{claim}

\begin{proof}
Fix an index $i\in [m]$.  Let $z_j=z^{(i)}(P_j)$ and $d_j=d^{(i)}(P_j)$. Thus $z_j$ ($d_j$) is the number of vectors from $P_j$ whose $i$-th coordinate is $0$  ($1$). Let $n_j=\vert P_j\vert$ and $p_j=\frac{d_j}{n_j}$,   $j\in \{1,\dots, k\}$. 
Since we assume that  ${\mathscr E}$ happened, each vector from 
$S_j$ is equally likely to be any vector in $P_j$. That is, for 
every vector $v\in S_j$, we have that 
\begin{equation}
\label{eqn:pi1}
 \probability{v[i]=1}{{\mathscr E}}
 =\frac{d_j}{n_j}=p_j \qquad\mbox{ and } \qquad  \probability{v[i]=0}{{\mathscr E}}
 =\frac{z_j}{n_j}=1-p_j.  
\end{equation} 
Let $\bfp=(p_1,\ldots,p_{k})$.  We define  $k$-tuple   $\bfy=(n_1, \dots, n_k)$  
and  claim that for every $\bfb=(b_1,\ldots,b_{k})\in R_i$, $d^{\bfy}(\bfb,\bfp)=g_i(\bfb)$. Indeed,
\begin{eqnarray}
d^{\bfy}(\bfb,\bfp)&=&\sum_{j=1}^{k} n_j \vert b_j - p_j \vert \nonumber\\
&=&\sum_{j\in I_b} n_j (1 - p_j)+\sum_{j\in [k]\setminus I_b} n_j  \cdot p_j \nonumber\\
&=& \sum_{j\in I_b} z_j+\sum_{j\in [k]\setminus I_b} d_j  \nonumber\\
&=&g_i(\bfb). \label{eqn:wgi1}
\end{eqnarray}

%

For each set $S_j=\{v_1, \dots, v_r\}$, $1\leq j\leq k$   and $1\leq q\leq r$, we define random variable $L_{j}^{q}$ which is $1$ when $v_{q}[i]=1$ and $0$ otherwise.
Let us denote by  $X_{j}^{q}$ the random variable $L_{j}^{q}\mid {\mathscr E}$. 
By \eqref{eqn:pi1}, we have that $X_{j}^{q}\sim B(p_j)$ for all $1 \leq j \leq k$ and $1 \leq q \leq r$. 
Let $\bfQ = (Q_1, \ldots Q_k)$ be the $k$-tuple of random variables,  where $Q_j = \frac{1}{r}\sum_{q=1}^r X_j^{q}$.
From the definitions of $z^{(i)}(S_j)$, $d^{(i)}(S_j)$ and $X_j^q$, we have that 
\begin{eqnarray}
z^{(i)}(S_j)=\sum_{q=1}^r (1-X_{j}^{q})  \quad \mbox{ and } \quad 
d^{(i)}(S_j)=\sum_{q=1}^r X_{j}^{q}. \label{eqn:zd}
\end{eqnarray}

Let $\bfw=(w_1,\ldots,w_k)$. For   $\bfb\in R_i$, we 
  upper bound $d^\bfw(\bfb,\bfQ)$ in terms of $f_i(\bfb)$. 
\begin{eqnarray}
 d^\bfw(\bfb,\bfQ)&=&\sum_{j=1}^{k} w_j \vert b_i - Q_j \vert \nonumber\\
&\leq & \sum_{j\in I_b} w_j (1 - Q_j)+\sum_{j\in [k]\setminus I_b} w_j  \cdot Q_j\nonumber\\
&\leq &  \sum_{j\in I_b} w_j \cdot \left(1- \frac{1}{r}\sum_{q=1}^r X_{j}^{q}\right) + \sum_{j\in [k]\setminus I_b} w_j \cdot \left(\frac{1}{r}\sum_{q=1}^r X_{j}^{q}\right)\nonumber\\
&\leq & \frac{1}{ r} \left(\sum_{j\in I_b} w_j \cdot \left(\sum_{q=1}^r (1-X_{j}^{q})\right) + \sum_{j\in [k]\setminus I_b} w_j \cdot \left(\sum_{q=1}^r X_{j}^{q}\right)\right)\nonumber\\
&\leq & \frac{1}{r} \left( \sum_{j\in I_b} w_j \cdot z^{(i)}(S_j) + \sum_{j\in [k]\setminus I_b} w_j \cdot d^{(i)}(S_j)\right)\qquad\qquad\qquad \mbox{(by \eqref{eqn:zd})}\nonumber\\
&\leq& \frac{1}{r} f_i(\bfb). \qquad\qquad\qquad\qquad (\mbox{by \eqref{eqn:fi}}) \label{eqn:ficorrect}
\end{eqnarray}

Let  
\[
\bfrho=\argmin_{\bfx\in R_i} d^\bfw(\bfx,\bfQ)
.\]
By \eqref{eqn:ficorrect} and by the definition of the vector set $C=\{\bfc_1,\ldots,\bfc_{k}\}$, we have that $\bfrho=(\bfc_1[i],\ldots,\bfc_{k}[i])$. 
 
 We also define $k$-tuples 
 \[
 \bfu=\argmin_{\bfx\in R_i} d^\bfw(\bfx,\bfp) \text{ and } \bfu^*=\argmin_{\bfx\in R_i} d^\bfy(\bfx,\bfp)
.\]

Thus the minimum weighted by $\bfw$ distance  $d_{min} $ from $\bfp$ to a $k$-tuple from $R_i$ is equal to $d^\bfw( \bfu,\bfp)$.    Let $D$ be the random variable which is a minimum weighted by $\bfw$ distance from $\bfrho$ to $\bfp$. By Lemma~\ref{lem:samplingMain}, 
\begin{equation}
\label{eq:samplingMain}
\expectation{d^\bfw(\bfrho,\bfp)}{\mathscr E}=  \expect[D] \leq (1+\epsilon)d_{min}=  (1+\epsilon) d^\bfw( \bfu,\bfp) . \end{equation}

Since $(\bfc_1^*[i],\ldots,\bfc_{k}^*[i])\in R_i$, 
and because $d^{\bfy}(\bfb,\bfp)=g_i(\bfb)$  for each $\bfb\in R_i$ (by \eqref{eqn:wgi1}), we have that 
\begin{equation}
\label{eqn:dminapp2}
d^\bfy(\bfu^*,\bfp) \leq g_i(\bfc_1^*[i],\ldots,\bfc_{k}^*[k])= V_i. 
\end{equation}

Finally we upper bound $\expectation{Y_i}{\mathscr E}$. 
\begin{eqnarray*}
\expectation{Y_i}{\mathscr E}&=&\expectation{g_i(\bfrho)}{\mathscr E}\\
&=&\expectation{d^\bfy(\bfrho,\bfp)}{\mathscr E} \qquad\qquad\qquad\quad(\mbox{by \eqref{eqn:wgi1}})\\
&\leq&c\cdot \expectation{d^\bfw(\bfrho,\bfp)}{\mathscr E}\qquad\qquad\quad (\mbox{because } n_j\leq cw_j \mbox{ for all } j\in [k])\nonumber\\
&\leq&c\cdot (1+\varepsilon)\cdot d^\bfw( \bfu,\bfp) \qquad\qquad (\mbox{by  \eqref{eq:samplingMain})}\nonumber\\
&\leq&c\cdot (1+\varepsilon)\cdot d^\bfw(\bfu^*,\bfp) \qquad\qquad (\mbox{by the choice of $ \bfu$})\nonumber\\
&\leq& (1+\varepsilon)(1+\delta) d^{\bfy}(\bfu^*,\bfp)  \qquad\quad (\mbox{because }  cw_j \leq (1+\delta) n_j \mbox{ for all } j\in [k]) \\
&\leq& (1+\varepsilon)(1+\delta) V_i. \qquad\qquad\qquad (\mbox{by \eqref{eqn:dminapp2}})
 \end{eqnarray*}
 This completes the proof of the claim.
\end{proof}

By Claim~\ref{claim:yiexp}, the fact that $V=\sum_{i=1}^{m} V_i$ and by the linearity of expectation,  we have that 
\begin{eqnarray*}
\expectation{Y}{\mathscr E}\leq (1+\varepsilon)(1+\delta) V .\label{eqn:expYgivenE}
\end{eqnarray*}
Combined with   Markov's inequality (Proposition~\ref{prop:markov}), this implies that 

\[
\probability[1]{Y \geq (1+\varepsilon)^2(1+\delta) V}{  {\mathscr E}} \leq \frac{1}{1+\varepsilon}.
\]
Therefore 
\begin{equation}
\probability[1]{Y \leq (1+\varepsilon)^2(1+\delta) V}{{\mathscr E}}\geq \frac{\varepsilon}{1+\varepsilon}. \label{eqn:YEhighPr}
\end{equation}
Finally, 
\begin{eqnarray*}
\Pr(Y\leq (1+\varepsilon)^2(1+\delta) V) &\geq& \Pr\left ((Y \leq (1+\varepsilon)^2(1+\delta)V) \cap {\mathscr E}\right)\\
&=&\probability[1]{Y \leq (1+\varepsilon)^2(1+\delta) V}{{\mathscr E}}  \cdot \Pr({\mathscr E}) \\
&\geq & \frac{\varepsilon}{1+\varepsilon} \cdot \beta^{r\cdot k}. \qquad (\mbox{by \eqref{eqn:YEhighPr} and  ~\eqref{obs:prE}})\\
\end{eqnarray*}
This completes the 
proof of the lemma.  
\end{proof}



\section{Non-irreducible instances and extendable solutions}\label{sec:non-reduceb}

For the algorithm of {\sc $k$-means clustering} Kumar et al.~\cite{KumarSS10} used the notion of 
{\em irreducible} instances. We introduce a similar notion for \rclustering.  

\subsection{Non-irreducible instances}
\begin{definition}
Let $J=(X,k,{\cal R})$ be an instance of \rclustering, $j\in \{2,\ldots,k\}$ and $\alpha>0$. We say that $J$ is 
$(j,\alpha)$-irreducible if $\opt_{j-1}(J)\geq (1+\alpha)\opt_j(J)$. 
\end{definition}


The following property of irreducible instances plays an important role in our algorithm. 
\begin{lemma}
\label{lem:irr}
Let $J=(X,k,{\cal R})$ be an instance of \rclustering. For  every
 $0<\epsilon\leq 4$ and $0<\alpha\leq \frac{\epsilon}{8k}$, the following holds. Let 
 \[
\widehat{k}=\begin{cases}
1, \text{ if } J \text{ is  not  $(i,\alpha)$-irreducible for all } i\in\{2,\dots, k\},\\
\max\{i~|~J \text{ is  $(i,\alpha)$-irreducible}\},   \text{ otherwise}.
\end{cases}
 \]
Then $\opt_{\widehat{k}}(J)\leq (1+\frac{\epsilon}{4})\opt(J)$.
%
\end{lemma}
\begin{proof}
By the choice of $\widehat{k}$, for every $\widehat{k}\leq i <k$, we have that   $\opt_{i}(J)\leq (1+\alpha)\opt_{i+1}(J)$. 
 Thus   $\opt_{\widehat{k}}(J)\leq (1+\alpha)^{{k}}\opt_{k}(J)=(1+\alpha)^{{k}}\opt(J)$. Since
\begin{eqnarray*}
 (1+\alpha)^{{k}}&\leq &(1+\frac{\epsilon}{8k})^{{k}}  =   \sum_{i=0}^{ k} \binom{ k}{i}\left(\frac{\epsilon}{8k}\right)^i \\
&\leq &  \sum_{i=0}^{ k} \left(\frac{\epsilon}{8}\right)^i  
 =   \left(1+\frac{\epsilon}{8}\sum_{i=0}^{k-1} \left(\frac{\epsilon}{8}\right)^{i}\right) \\
&\leq&  \left(1+\frac{\epsilon}{8}\sum_{i=0}^{k-1} \left(\frac{1}{2}\right)^{i}\right)  \leq   \left(1+\frac{\epsilon}{4}\right),    \qquad\qquad (\mbox{since }\epsilon\leq 4)  
\end{eqnarray*} 
we have that 
 $\opt_{\widehat{k}}(J)\leq  \left(1+\frac{\epsilon}{4}\right)\opt(J)$.
\end{proof}

Due to Proposition~\ref{prop:smallsolution} and Lemma~\ref{lem:irr},  to obtain 
an $(1+\epsilon)$-approximate solution to  \rclustering, it is sufficient to learn how to approximate irreducible instances. Indeed, let 
$J=(X,k,{\cal R}=\{R_1,\ldots,R_k\})$ be an instance to \rclustering and suppose that for $\epsilon>0$ and  $\alpha\leq \frac{\epsilon}{8k}$, instance 
$J $ is not $(k,\alpha)$-irreducible.  
Then   for the index $\widehat{k}$ defined in Lemma~\ref{lem:irr}, we have that 
$\widehat{k}<k$ and $\opt_{\widehat{k}}(J)\leq  \left(1+\frac{\epsilon}{4}\right)\opt(J)$.  The definition of $\opt_{\widehat{k}}(J)$, implies that 
\[\opt_{\widehat{k}}(J)=\min\{ \opt(X,\widehat{k},\proj{\RR}{I})~ |~ I\subseteq \{1,\dots, k\} \text{ and } |I|=\widehat{k}\}.\]

%
We can make a guess for the value of $\widehat{k}$ and then guess an index subset  $I\subset \{1,\dots, k\}$ of size $\widehat{k}$. For each of the $(k-1)\cdot 2^{k-1}$ guesses of  $\widehat{k}$ and $I$, we 
form  instance $J'=(X, \widehat{k} ,\proj{\RR}{I})$. We know that for at least one of our guesses, we will have an $(\widehat{k}, \alpha)$-irreducible instance  $J'$ with  $\opt(J')=\opt_{\widehat{k}}(J)$.
%
%
%
 By Lemma~\ref{lem:irr} and Proposition~\ref{prop:smallsolution}, any $(1+\epsilon/4)$-approximate solution to $J'$ is extendable in linear time to a  $(1+\epsilon/4)^2$-approximate solution (and hence a  $(1+\epsilon)$-approximate solution) to $J$. As a result if $J$ 
is not $(k,\alpha)$-irreducible, then  a $(1+\epsilon/4)$-approximate solution to $J'$ will 
bring us  to  a $(1+\epsilon)$-approximate solution   to $J$.  
Hence everything boils down to the approximation of $(k,\alpha)$-irreducible instances.

\subsection{Extendable solutions}

By Lemma~\ref{lem:samplelargecluster},  if there is solution such that  the sizes of all  corresponding clusters
are constant fractions of the number of the input vectors, then   sampling  produces a good approximate solution.
 However, there is no guarantee that such a favorable  condition will occur.   
 To overcome this we sample vectors 
for large clusters and then identify some vectors in the input which we can {\em safely} delete 
and make the next largest remaining cluster  a constant fraction of the rest of the vectors.    
Towards that we need the following definition.


\begin{definition}[$\delta$-extension of a solution]
Let $J=(X,k,{\RR})$ be an instance of \rclustering\ and $\delta\geq 0$. Let $B\subseteq X$ and $C_1\subseteq \{0,1\}^m$ be a set of $k_1$ vectors for some $k_1\leq k$. The pair $(C_1,B)$ is called {\em $\delta$-extendable for instance $J$} if there is a vector set $C_2\subseteq \{0,1\}^m$ of size $k-k_1$  such that  $C_1\cup C_2$  satisfies $\RR$  and $\Hdist(B,C_1)+\Hdist(X\setminus B, C_1\cup C_2)\leq (1+\delta)\opt(J)$.    
We say that $C_2$ is a {\em $\delta$-extension  of $(C_1,B)$}.
\end{definition}

In particular, if $(C_1,B)$ is $\delta$-extendable for $J$, then there is a  set $C\supseteq C_1$ 
such that $C$ is a  solution to $J$  with cost  at most $(1+\delta)\opt(J)$ even when $B$ is assigned to the clusters corresponding to the center vectors in $C_1$. This implies that after we find such a set $C_1$, 
it is safe to delete the vectors $B$ from the input set of vectors. If $C_2$ is a $\delta$-extension of $(C_1,B)$, then 
there exist index subset $I\subseteq\{1,\dots, k\}$ of size $\vert C_1\vert$  
such that $C_1$ satisfies $\proj{\RR}{I}$ and $C_2$ satisfies $\RR(I,C_1)$. 
%
The proof of the next observation follows directly from the definition of   $\delta$-extention.

\begin{observation}
\label{obs:largeext}
Let $J=(X,k,{\cal R})$ be an instance of \rclustering,  $\delta'\geq \delta\geq 0$, and  $(C_1,B)$ be a $\delta$-extendable pair for $J$. Then
 \begin{itemize}
 \item 
 $(C_1,B)$ is also $\delta'$-extendable for $J$, 
 \item  if $C_2$ is a $\delta$-extension of $(C_1,B)$, then $C_2$ is also a $\delta'$-extension of $(C_1,B)$, and 
\item for each $B'\subseteq B$, pair  $(C_1,B')$  is $\delta$-extendable for $J$. 
\end{itemize}
\end{observation}

Let $(C_1,B)$ be a pair which is  $\delta$-extendable for $J$ and $C\supseteq C_1$ be such that 
$C\setminus C_1$ is a  $\delta$-extension of $(C_1,B)$.  While the set  $C\setminus C_1$ is not known to us and we do not know yet how to compute it, as we will see in the following lemma, we can proceed successfully even  if only the minimum Hamming distance $t$ between vectors in $C_1$ and $C\setminus C_1$ is known.
%
We show that if we know $t$, then 
we can find a set of vectors $B'\subseteq X\setminus B$ such that $B'$ is not only   {  safe} to delete 
but  the number of vectors in $X\setminus (B\cup B')$ which will be assigned to clusters corresponding to $C_1$ (in the solution $C$)
is at most a constant fraction of the   number of vectors in $X\setminus (B\cup B')$. In other words the number of vectors  which will be assigned to clusters corresponding to $C\setminus C_1$ will be at least a  constant fraction of the   number of vectors in $X\setminus (B\cup B')$. This information is of crucial importance; it yields that from the remaining set of vectors at least one of the clusters assigned to  $C\setminus C_1$ is large, and  thus could be found by sampling.

\begin{lemma}
\label{lem:safetoremove}
Let $J=(X,k,{\cal R} )$ be an instance of \rclustering\ and $\delta \geq 0$. Let  $(C_1,B)$, where $B\subseteq X$ and $C_1\subseteq \{0,1\}^m$, 
be a {$\delta$-extendable pair  for $J$}. Let $C_2$ be a $\delta$-extension of $(C_1,B)$   and $t=\min \{\hdist(\bfc,\bfc')\colon \bfc\in C_1, \bfc'\in C_2\}$. 
Let $(Z_1, Z_2)$ be a partition of $X\setminus B$  such that 
 $\Hdist(X\setminus B, C_1\cup C_2)=
\Hdist(Z_1, C_1) + \Hdist(Z_2,C_2)$. 
Let
 $B'=\bigcup_{\bfc\in C_1}\BB(\bfc,t/2)\cap (X\setminus B)$.    
Then 
the following conditions hold. 
\begin{itemize}
\item[$(i)$] $B'\subseteq Z_1$. Moreover,   $B'$ consists of the first $\vert B'\vert$ vectors of $X\setminus B$ in the ordering  
according to the non-decreasing distance $\hdist(\bfx,C_1)$ (where $\bfx\in X\setminus B$). 
\item[$(ii)$] $\Hdist(X\setminus (B\cup B'),C_1\cup C_2)=\Hdist(Z_1\setminus B',C_1)+\Hdist(Z_2,C_2)$. Moreover,   $(C_1,B\cup B')$ is $\delta$-extendable for $J$ and $C_2$ is a $\delta$-extension of $(C_1,B\cup B')$.   
\item [$(iii)$] If $J$ is $(k,5\delta')$-irreducible for some $\delta'\geq \delta$, then  $\vert Z_2\vert \geq \left(\frac{\delta'}{1+\delta'}\right) \vert X\setminus (B\cup B') \vert$.
If in addition,     $\vert B'\vert \leq \frac{\vert X\setminus B\vert}{2}$, then 
$\vert Z_2 \vert \geq  \left( \frac{\delta'}{2(1+\delta')}\right) \vert X \setminus B \vert.$
\end{itemize} 
\end{lemma}
\begin{proof}
We start with  $(i)$. 
 Since $t=\min \{\hdist(\bfc,\bfc')\colon \bfc\in C_1, \bfc'\in C_2\}$, for any vector $\bfx$ in $\bigcup_{\bfc\in C_1}\BB(\bfc,t/2)$, the value 
 $\hdist(\bfx,C_2)$ is strictly greater  than $t/2$. Because $\Hdist(X\setminus B, C_1\cup C_2)=
\Hdist(Z_1, C_1) + \Hdist(Z_2,C_2)$, we conclude that $B'\subseteq Z_1$. 
 Now if we order the vectors  of $X\setminus B$ according to their   distances to $C_1$ (the smallest distance comes first and ties are broken arbitrarily), then the first $|B'|$ vectors in this ordering are within distance strictly less than  $t/2$ from $C_1$, while for any other vector $\bfx \not\in B'$, $\hdist(\bfx,C_1)\geq t/2$. This concludes the proof of  $(i)$.

  

Before proceeding with  $(ii)$ and  $(iii)$, we observe the following.
First,  since $C_2$ is a $\delta$-extension of $(C_1,B)$, we have that 
\begin{equation*}
\label{eqn:firstobs}
\Hdist(X,C_1\cup C_2)\leq \Hdist(B,C_1)+\Hdist(X\setminus B,C_1\cup C_2)\leq (1+\delta)\opt(J).
\end{equation*}
By the definition of sets $Z_1$ and $Z_2$, we have
that 
 \begin{eqnarray*}
\Hdist(B,C_1)+\Hdist(X\setminus B,C_1\cup C_2) &=& \Hdist(B,C_1)+  \Hdist( Z_1,C_1)+\Hdist(Z_2,C_2)\\ &=&\Hdist(B\cup Z_1,C_1)+\Hdist(Z_2,C_2).
\end{eqnarray*}

\begin{equation}
\label{eqn:firstobs}
\mbox{Thus, }\qquad\Hdist(X,C_1\cup C_2)\leq \Hdist(B\cup Z_1,C_1)+\Hdist(Z_2,C_2)\leq (1+\delta)\opt(J).
\end{equation}

We need one more observation and its proof follows directly from the definition of sets $Z_1$ and $Z_2$.
 \begin{observation}
\label{obs:sameZ2}
For every $\bfx\in Z_1$, $\hdist(\bfx, C_1\cup C_2)=\hdist(\bfx, C_1)$ 
and for every $\bfy\in Z_2$, $\hdist(\bfy, C_1\cup C_2)=\hdist(\bfy, C_2)$. 
\end{observation}

Now we prove condition $(ii)$. 
 First, by Observation~\ref{obs:sameZ2} and the fact that $B'\subseteq Z_1$, 
we have that 
\begin{equation}
\label{eqn:condition2sec}
\Hdist(X\setminus (B\cup B'),C_1\cup C_2)=\Hdist(Z_1\setminus B',C_1)+\Hdist(Z_2,C_2).
\end{equation}
To prove that $(C_1,B\cup B')$ is $\delta$-extendable for $J$ and that $C_2$ is a $\delta$-extension of $(C_1,B\cup B')$, we note that
\begin{eqnarray*}
\Hdist(B\cup B',C_1) &+&\Hdist(X\setminus (B\cup B'),C_1\cup C_2)\\ 
 &=& \Hdist(B\cup B',C_1)+\Hdist(Z_1\setminus B',C_1)+ \Hdist(Z_2,C_2) \qquad\quad (\mbox{by }\eqref{eqn:condition2sec}) \\
&\leq&(1+\delta) \opt(J).  \qquad\qquad\qquad\qquad\qquad\qquad(\mbox{by }\eqref{eqn:firstobs} 
\mbox{ and because }B'\subseteq Z_1) 
\end{eqnarray*} 


Now we prove condition $(iii)$ of the lemma. 
Let $C_2=\{\bfc_1,\ldots,\bfc_{\ell}\}$, where $\ell=k-\vert C_1\vert$. Let   $Y_{1}\uplus\cdots \uplus Y_{\ell}$ 
be a  partition of $Z_2$ such that 
\begin{equation*}
\Hdist(Z_2,C_2)=\sum_{j=1}^{\ell} \sum_{\bfy\in Y_j} \hdist(\bfc_j,\bfy).
\end{equation*}
Let vector   
$\bfc\in C_1$ and index $r\leq \ell$ be such  that $t=\hdist(\bfc,\bfc_r)$. 

\begin{claim}
\label{claim:yrbound00}
If $J$ is $(k,5\delta')$-irreducible for some $\delta'\geq \delta$, then $\vert Z_1\setminus B' \vert \leq \frac{1}{\delta'}\cdot \vert Y_r \vert$.
\end{claim}
\begin{proof}
For the sake of contradiction assume that $\vert Z_1\setminus B' \vert > \frac{1}{\delta'}\cdot \vert Y_r \vert$. 
This implies that  $\Hdist(Z_1\setminus B',C_1)\geq \frac{t}{2\delta'} \vert Y_r\vert$. Thus  
\begin{equation}
\label{eqn:change000}
t \vert Y_r\vert 
  \leq 2\delta' \Hdist(Z_1\setminus B',C_1) \leq 2\delta' \Hdist(Z_1,C_1).
\end{equation}
Now we upper bound $\Hdist(X,C_1\cup (C_2\setminus \{c_r\}))-\Hdist(X,C_1\cup C_2)$, by considering the the change in cost 
when we reassign the vectors in $Y_r$ to the cluster with center $\bfc$. 
\begin{eqnarray}
\Hdist(X,C_1\cup (C_2\setminus \{\bfc_r\}))&-&\Hdist(X,C_1\cup C_2)\nonumber\\
&\leq& 
\sum_{\bfy\in Y_r} (\hdist(\bfy,\bfc)- \hdist(\bfy,\bfc_{r}))\nonumber\\
&\leq& \sum_{\bfy\in Y_r} \hdist(\bfc,\bfc_{r}) \quad (\mbox{by the triangle inequality})\nonumber\\
 &=&  \vert Y_r\vert\cdot t  
 \leq  2\delta' \cdot \Hdist(Z_1,C_1) \qquad\qquad\qquad (\mbox{by }\eqref{eqn:change000}) \label{eqn:change111}
\end{eqnarray}
   
But then
\begin{eqnarray*}\opt_{k-1}(J) & \leq &
\Hdist(X,C_1\cup (C_2\setminus \{\bfc_r\}))\\ &\leq& \Hdist(X,C_1\cup C_2)+ 2\delta' \Hdist(Z_1,C_1) \qquad(\mbox{by }\eqref{eqn:change111})\\
&\leq& (1+\delta)\cdot \opt(J)+2\delta' (1+\delta)\cdot \opt(J) \quad(\mbox{by }\eqref{eqn:firstobs})\\
&\leq& (1+\delta+4\delta')\cdot \opt(J)\qquad (\mbox{since }\delta\leq 1)\\
&\leq& (1+5\delta')\cdot  \opt(J). \qquad (\mbox{since }\delta\leq \delta')
\end{eqnarray*}
This contradicts the assumption that $J$ is $(k,5\delta')$-irreducible.  
\end{proof}
Since $Y_r\subseteq Z_2$,  by Claim~\ref{claim:yrbound00}, we have that $\vert Z_1\setminus B' \vert \leq \frac{1}{\delta'}\cdot \vert Z_2 \vert$. 
Sets $Z_1$ and $Z_2$ form a partition of $X \setminus  B$ and because $B'\subseteq Z_1$, we have that $X\setminus (B\cup B')=(Z_1\setminus B') \cup Z_2$.
This implies that $|X\setminus (B\cup B')|\leq |Z_2| (1+1/\delta')$, and hence
\[
\vert Z_2\vert \geq \frac{\delta'}{1+\delta'} \cdot \vert X \setminus (B\cup B')\vert .
\] 
Finally, we prove the last condition in $(iii)$.  If  $\vert B'\vert \leq \frac{\vert X\setminus B\vert}{2}$,  then
\begin{eqnarray*}
\vert  Z_2 \vert &\geq & \frac{\delta'}{1+\delta'} \cdot \vert X \setminus (B\cup B')  \vert  \\
&\geq &\frac{\delta'}{1+\delta'}   \cdot \left(\vert X \setminus B\vert- \vert B' \vert \right)\\
&\geq &  \frac{\delta'}{2(1+\delta')}  \cdot \vert X \setminus B\vert.
\end{eqnarray*}
This completes the proof of the lemma.
\end{proof}

Due to Lemma~\ref{lem:safetoremove}, once we have a partial solution $C_1$, we can identify a set $B'$ of vectors 
such that the number of vectors in the largest cluster corresponding to $C\setminus C_1$ (here $C$ is a {\em good} solution which contains $C_1$) is at least a constant fraction of the remaining vectors. This allows us to further use Lemma~\ref{lem:samplelargecluster} as formally 
explained in Lemma~\ref{lem:findmore}. To make the statements in Lemma~\ref{lem:findmore} easier, we make use of   one of the {\em best} 
$\delta$-extension of $(C_1,B)$ as derived in the following lemma. 

\begin{definition}[Good $\delta$-extension]
\label{def:goodetn}
For   instance $J=(X,k,{\cal R}=\{R_1,\ldots,R_m\})$ of \rclustering, $\delta \geq 0$, $B\subseteq X$ and $C_1\subseteq \{0,1\}^m$ of size $k'$, 
a $\delta$-extension $C_2$ of $(C_1,B)$ is a {\em good $\delta$-extension of $(C_1,B)$} if there is a partition $Z_1\uplus Z_2=X\setminus B$ and   index set  $I \subseteq \{1,\dots, k\}$ of size $k'$  such that 

\begin{itemize}
\item 
 $\Hdist(X\setminus B,C_1\cup C_2)=\Hdist(Z_1,C_1)+\Hdist(Z_2,C_2)$, 
\item 
$C_1$ satisfies  $\proj{\RR}{I}$, and 
\item $C_2$ is an optimal  solution to $J'=(Z_2, k-k',\RR'=\RR(I,C_1))$. \end{itemize}
We will refer to $J'$ as to  the $C_2$-optimal  reduced  instance. 


\end{definition}

\begin{lemma}
\label{lem:bestC2}
Let $J=(X,k,{\cal R}=\{R_1,\ldots,R_m\})$ be an instance of \rclustering\ and $\delta \geq 0$. Let $B\subseteq X$ and $C_1\subseteq \{0,1\}^m$ be a set of $k'$ vectors for some $k'\leq k$ such  that $(C_1,B)$ is {$\delta$-extendable for $J$}.
Then there is a good $\delta$-extension $C_2$ of $(C_1,B)$.
\end{lemma}
\begin{proof}
Among all  $\delta$-extensions of $(C_1,B)$, let $C_2$ be a $\delta$-extension of $(C_1,B)$ such that 
$\Hdist(X\setminus B,C_1\cup C_2)$ is minimized. We prove that  $C_2$ is the required $\delta$-extension of $(C_1,B)$. 

Let 
\[\eta=\Hdist(X\setminus B,C_1\cup C_2).\] 
Since $C_2$ is a $\delta$-extension of $(C_1,B)$, there is a set  
$I\subseteq \{1,\dots, k\}$ of size $k'$ 
 and a partition $Z_1\uplus Z_2$ of $X\setminus B$ 
such that $C_1$ and $C_2$ satisfy families of relations $\proj{\RR}{I}$ and $\RR(I,C_1)$ correspondingly, and 
%
\[\eta =\Hdist(Z_1,C_1)+\Hdist(Z_2,C_2)\leq (1+\delta)\opt(J).\] 
 By the choice of $C_2$,  for every  $\delta$-extension $C_2'$ of $(C_1,B)$, we have that $\Hdist(X\setminus B,C_1\cup C_2')\geq \eta$. Assume, for the sake of contradiction,  that $C_2$ is not an optimal  solution to $(Z_2,k-k',\RR(I,C_1))$. Let $C_2^*$ be an optimal  solution to $(Z_2,k-k',\RR(I,C_1))$. Then   $\Hdist(Z_2,C_2^*)<\Hdist(Z_2,C_2)$. Moreover,  $C_1\cup C_2^*$ 
satisfies $\RR$ and hence is a solution to $J$. Since $\Hdist(Z_2,C_2^*)<\Hdist(Z_2,C_2)$, we conclude that 
\[\Hdist(X\setminus B,C_1\cup C_2^*)\leq \Hdist(Z_1,C_1)+\Hdist(Z_2,C_2^*)<\Hdist(Z_1,C_1)+\Hdist(Z_2,C_2)=\eta.\] 
This contradicts the choice of $C_2$, which in turn, completes the proof of the lemma.  
\end{proof}

In order to state the next lemma, we need one more definition.
\begin{definition}[Set of $\kappa$-heavy clusters]
\label{def:heavyclusters}
Let $\kappa>0$, $X$ be a vector set and $X_1,\ldots,X_{\ell}$ be a clustering of $X$. We say that a subset of clusters ${\cal X}\subseteq \{X_1,\ldots,X_{\ell}\}$  is the set of \emph{$\kappa$-heavy clusters} if for every $Y \in \{\cal X\}$ and $Z \not\in {\cal X}$, 
 $\vert Y \vert > \kappa\vert Z\vert$. 
\end{definition}

The following lemma says that if we have an irreducible instance $J$ with  $\delta$-extendable pair $(B,C_1)$, 
then  it is possible to construct a larger extendable pair   by adding to $C_1$ a set of ``approximate'' centers of heavy clusters from optimal clustering of instance  $J'$ obtained from $J$ by ``subtracting'' $(C_1, B)$.

\begin{lemma}
\label{lem:findmore}
Let $\delta\geq 0$ and $0\leq \alpha \leq 1$  and let $J=(X,k,{\cal R}=\{R_1,\ldots,R_m\})$ be a $(k,5\delta')$-irreducible instance of \rclustering for some $\delta'\geq \delta$. 
Let pair $(C_1,B)$,  $B\subseteq X$, $C_1\subseteq \{0,1\}^m$, $\vert C_1\vert <k$, be  $\delta$-extendable for $J$. Let $C_2=\{\bfc_1^*,\ldots,\bfc_\ell^*\}$ be a good $\delta$-extension of $(C_1,B)$ and $J'=(Z_2,\ell=k-\vert C_1\vert,\RR'=\RR(I,C_1))$ be the corresponding 
$C_2$-optimal  reduced  instance.
%
%
(The existence of such $C_2$   is guaranteed by Lemma~\ref{lem:bestC2}.) 
Let also  $X_1,\ldots,X_{\ell}$ be the set of clusters corresponding  to  the solution $C_2$ of $J'$ and let $I'\subseteq\{1,\dots, \ell\}$, $|I'|=k'$,  be the set of indices of the 
 $\frac{k}{\alpha}$-heavy clusters from this set.

 Then for every solution 
$ C_2'=\{\bfc_i\colon i\in I'\}$ to 
$J_1'=(\bigcup_{i\in I'}X_i ,k',\proj{\RR'}{I'})$ satisfying condition  
\[\sum_{j\in I'}\Hdist(X_j,\{\bfc_j\})\leq (1+{\alpha})\cdot \sum_{j\in I'}\Hdist(X_j,\{\bfc_j^*\}),\]   the pair  $(C_1\cup C_2',B)$ is $(5\delta+4\alpha)$-extendable for $J$. 
\end{lemma}


\begin{proof}
Because $C_2$ is a   $\delta$-extension of $(C_1,B)$,
we have that 
\[\Hdist(B,C_1)+\Hdist(X\setminus B, C_1\cup C_2)\leq (1+\delta)\opt(J).\] 
The assumption that $C_2$ is a good $\delta$-extension of $(C_1,B)$, yields that 
 there are 
 $Z_1\uplus Z_2=X\setminus B$ and    $I \in \binom{[k]}{\vert C_1\vert}$  such that 
\begin{itemize}
\item 
 $\Hdist(X\setminus B,C_1\cup C_2)=\Hdist(Z_1,C_1)+\Hdist(Z_2,C_2)$, 
\item 
$C_1$ satisfies  $\proj{\RR}{I}$, and 
\item $C_2$ is an optimal  solution to $J'=(Z_2,\ell=k-\vert C_1\vert,\RR'=\RR(I,C_1))$. \end{itemize}

Hence  
\begin{eqnarray}
\Hdist(X,C_1\cup C_2)\leq \Hdist(B\cup Z_1,C_1)+\Hdist(Z_2,C_2)\leq (1+\delta)\opt(J). \label{eqn:firstobs1111}
\end{eqnarray}

For the ease of presentation we assume that $I'=\{1,\dots, k'\}$. 
Let $W_1=\bigcup_{j=I'} X_j$ and $W_2=Z_2\setminus W_1$. 
 Let $C_3=\{\bfc_{k'+1},\ldots,\bfc_{\ell}\}$ be an optimal  solution to $J_2'=(W_2, \ell-k', \RR'(I',C_2'))$. 
From the solution $C_3$ of $J_2'$, we define a solution $C_3'=\{\bfc'_{k'+1},\ldots,\bfc'_{\ell}\}$ to $J_2'$ as follows. 
For each  $i\in [m]$, we set
\[
 (\bfc'_{k'+1}[i],\ldots,\bfc'_{\ell}[i])   = 
\begin{cases}
    (\bfc^*_{k'+1}[i],\ldots,\bfc^*_{\ell}[i]), & \text{if }(\bfc_{1}[i],\ldots,\bfc_{k'}[i])=(\bfc^*_{1}[i],\ldots,\bfc^*_{k'}[i]),\\
    (\bfc_{k'+1}[i],\ldots,\bfc_{\ell}[i]),              & \text{otherwise}.
\end{cases}
\]

\begin{observation}
$C_3'$ is a solution to $J_2'$. 
\end{observation}
\begin{proof}
From the definition of $C_3'$, we have that for any $i\in [m]$, $(\bfc'_{k'+1}[i],\ldots,\bfc'_{\ell}[i]) \in \RR'(I',C_2')$. 
Hence $C_3'$ is a solution to $J_2$. 
\end{proof}

Let us remind that we define $\RR'=\RR(I,C_1)$. Because 
  $C_1\in \proj{\RR}{I}$, $C_2'\in \proj{\RR'}{I'}$ and $C_3'\in \RR'(I',C_2')$, we have that  
$C_1\cup C_2'\cup C_3'$ is a solution to $J$. We will prove that 
$C_3'$ is indeed a $(5\delta+4\alpha)$-extension of $C_1\cup C_2'$.  
Towards that 
we 
 define

\begin{eqnarray}
\Delta_1&=&\sum_{j\in I'}\sum_{\bfx\in X_j}\hdist(\bfx,\bfc^*_j)=\sum_{j\in I'} \Hdist(X_j,\{\bfc_j^*\})  \label{eqn0O1}, 
 \end{eqnarray}
 and
\begin{eqnarray}
\Delta_2&=&\sum_{j\in [\ell]\setminus I'}\sum_{\bfx\in X_j}\hdist(\bfx,\bfc^*_j)=\sum_{j\in [\ell]\setminus I'} \Hdist(X_j,\{\bfc_j^*\})\label{eqn0O2}
\end{eqnarray}
By  counting the number of mismatches at each of the coordinates of vectors in $W_1$ with its corresponding center in $\{\bfc_1^*,\ldots,\bfc_{k'}^*\}$,  
we get that 
\begin{equation}
\label{eqn:rowwise1}
\sum_{r=1}^{m}\sum_{j\in I'} \sum_{\substack{\bfx\in X_j : \bfx[r]\neq \bfc^*_j[r]}}1=\sum_{j\in I'} \Hdist(X_j,\{\bfc_j^*\})=\Delta_1,
\end{equation}
where the last equality follows from \eqref{eqn0O1}. 
%
%
Since $C$ is an optimal  solution to $J'$ with corresponding clusters $X_1,\ldots, X_{\ell}$, we have that 
$\opt(J')=\Delta_1+\Delta_2$. 
Combining \eqref{eqn0O1} with 
the  assumption 
$\sum_{j\in I'}\Hdist(X_j,\{\bfc_j\})\leq (1+{\alpha})(\sum_{j\in I'}\Hdist(X_j,\{\bfc_j^*\}))$,  
we have that 
\begin{eqnarray}
\Hdist(W_1, C_2') \leq \sum_{j\in I'}\Hdist(X_j,\{\bfc_j\}) 
\leq (1+{\alpha})\Delta_1. \label{eqn:0j1}
\end{eqnarray}  

By arguments similar to the reasoning for \eqref{eqn:rowwise1} and  by \eqref{eqn:0j1},  we obtain the following. 

\begin{equation}
\label{eqn:rowwise2}
\sum_{r=1}^{m}\sum_{j\in I'} \sum_{\substack{\bfx\in X_j:\bfx[r]\neq \bfc_j[r]}}1=\sum_{j\in I'} \Hdist(X_j,\{\bfc_j\})=(1+\alpha)\Delta_1. 
\end{equation}

We claim that
\begin{claim}
\label{claim:secpart}
$\Hdist(W_2,C_3')\leq \Delta_2+(3\alpha\cdot \Delta_1)$
\end{claim}

Before proving the claim, let us show how it  concludes the proof of the lemma. 
By \eqref{eqn:0j1} and Claim~\ref{claim:secpart}, we have that 
\begin{eqnarray}
\Hdist(W_1,C_2') +\Hdist(W_2,C_3') &\leq&  (1+\alpha) \Delta_1 + \Delta_2+(3\alpha\cdot \Delta_1)\nonumber \\
&\leq& (1+4\alpha) (\Delta_1+\Delta_2)\nonumber\\
&\leq& (1+4\alpha) \Hdist(Z_2,C_2), 
\label{eqn:cond01}
\end{eqnarray}
where 
the last inequality follows from the fact that $C_2=\{\bfc_1^*,\ldots,\bfc_{\ell}^*\}$ is a solution to $J'$ 
with corresponding clusters $X_1,\ldots,X_{\ell}$ and equations \eqref{eqn0O1} and \eqref{eqn0O2}.

Now we prove that $C_3'$ is a $(5\delta+4\alpha)$-extension of $(C_1\cup C_2',B)$ for $J$. 
Towards that we upper bound $\eta=\Hdist(B,C_1\cup C_2')+\Hdist(X\setminus B,C_1\cup C_2'\cup C_3')$. 
\begin{eqnarray*}
\eta&\leq& \Hdist(B,C_1)+\Hdist(Z_1,C_1)+ \Hdist(W_1, C_2')+\Hdist(W_2, C_3')\\
&\leq&\Hdist(B,C_1)+\Hdist(Z_1,C_1)+ (1+4\alpha) \Hdist(Z_2,C_2) \qquad\qquad(\mbox{by }\eqref{eqn:cond01})\\
&\leq& (1+\delta) \opt(J)+(4\alpha) \Hdist(Z_2,C_2) \qquad\qquad\qquad\qquad\qquad(\mbox{by }\eqref{eqn:firstobs1111})\\
&\leq& (1+\delta) \opt(J)+(4\alpha) (1+\delta)\opt(J)\qquad\qquad\qquad(\mbox{by }\eqref{eqn:firstobs1111})\\
&\leq& (1+5\delta+4\alpha)\opt(J) \qquad\qquad\qquad\qquad\qquad\qquad (\mbox{since }\alpha\leq 1)
\end{eqnarray*}
This, subject to the proof of Claim~\ref{claim:secpart},  completes the proof of the lemma.

\begin{proof}[Proof of Claim~\ref{claim:secpart}]
We need one more technical claim.
\begin{claim}
\label{claim:aux1}
For any $j\in [\ell]\setminus I'$,  \[\sum_{\substack{r\in [m]:\bfc'_j[r]\neq \bfc_j^*[r]}} \vert X_j\vert\leq \frac{3\alpha}{k}\cdot \Delta_1.\]
\end{claim}
\begin{proof}
Fix $j\in [\ell]\setminus I'$. 
For any $r\in [m]$
 such that $\bfc'_j[r]\neq \bfc^*_j[r]$, by the definition of $(\bfc'_{k'+1}[r],\ldots,\bfc'_{k}[r])$, 
we have that $(\bfc_1[r],\ldots,\bfc_{k'}[r])\neq (\bfc^*_1[r],\ldots,\bfc^*_{k'}[r])$. That is, for any 
$r\in [m]$ such that $\bfc'_j[r]\neq \bfc^*_j[r]$, there is $j_r\in I'$ such that $\bfc_{j_r}[r]\neq \bfc^*_{j_r}[r]$. 
Define a map $f$ from $I_j=\{r\in [m]\colon \bfc'_j[r]\neq \bfc^*_j[r]\}$ to $I'$ as follows: 
$f(r)=j_r$, where $j_r$ is an arbitrary index in $I'$ such that $\bfc_{j_r}[r]\neq \bfc^*_{j_r}[r]$. 
We define sets 
\[
A[r]=\{\bfx\in X_{f(r)}~|~\bfx[r]\neq \bfc_{f(r)}[r]\},
\]
and 
\[
B[r]=\{ \bfx\in X_{f(r)}~|~\bfx[r]\neq \bfc^*_{f(r)}[r]\}.
\]

Since $\vert X_j\vert <
\frac{\alpha}{k} \vert X_{f(r)}\vert$ for every $r\in I_j$ (by the assumption that clusters with indices from $I'$ are $\frac{k}{\alpha} $-heavy), we have that 
\[
\vert X_j\vert \leq \frac{\alpha}{k} \vert X_{f(r)}\vert = \frac{\alpha}{k} \left(\sum_{\substack{ \bfx\in A[r]}} 1 + \sum_{\substack{ \bfx\in B[r]}} 1\right).
\]
Therefore, 
\begin{eqnarray*}
\sum_{\substack{r\in [m] :\bfc'_j[r]\neq \bfc_j^*[r]}} \vert X_j\vert &\leq& \sum_{\substack{r\in [m] :\bfc'_j[r]\neq \bfc_j^*[r]}} \frac{\alpha}{k} \left(\sum_{\substack{ \bfx\in A[r]}} 1 + \sum_{\substack{ \bfx\in B[r]}} 1\right)
\end{eqnarray*}
\begin{eqnarray*}
&\leq& \frac{\alpha}{k} \left( \left(\sum_{\substack{r\in [m] :\bfc'_j[r]\neq \bfc_j^*[r]}}  \sum_{\substack{ \bfx\in A[r]}} 1\right) + \left( \sum_{\substack{r\in [m] :\bfc'_j[r]\neq \bfc_j^*[r]}} \sum_{\substack{ \bfx\in B[r]}} 1\right)\right)\\
&\leq& \frac{\alpha}{k} \left( \left(\sum_{\substack{r\in [m] :\bfc'_j[r]\neq \bfc_j^*[r]}} \, \sum_{\substack{ \bfx\in A[r]}} 1\right) + \Delta_1\right) \qquad(\mbox{by }\eqref{eqn:rowwise1})\\
&\leq& \frac{\alpha}{k} \left( (1+\alpha)\Delta_1 + \Delta_1\right)\qquad\qquad\qquad\qquad(\mbox{by }\eqref{eqn:rowwise2})\\
&\leq &\frac{3\alpha}{k}\Delta_1. \qquad\qquad\qquad(\mbox{since }\alpha\leq 1)
\end{eqnarray*}
\end{proof}

Now we are ready to proceed with the proof  of Claim~\ref{claim:secpart}. We  bound

\begin{eqnarray*}
\Hdist(W_2,C_3') 
&\leq& \sum_{j\in [\ell]\setminus I'}\sum_{\bfx\in X_j}\hdist(\bfx,\bfc'_j)\\
&=&\sum_{j\in [\ell]\setminus I'}\sum_{\bfx\in X_j}\sum_{r=1}^{m}\vert \bfx[r]-\bfc'_j[r]\vert\\
&=&\sum_{j\in [\ell]\setminus I'}\sum_{\bfx\in X_j}\left( \sum_{\substack{r\in [m] :\bfc'_j[r]=\bfc_j^*[r]}}\vert \bfx[r]-\bfc'_j[r]\vert+\sum_{\substack{r\in [m] :\bfc'_j[r]\neq \bfc_j^*[r]}}\vert \bfx[r]-\bfc'_j[r]\vert\right)\\
&\leq&\sum_{j\in [\ell]\setminus I'}\sum_{\bfx\in X_j}\left(\hdist(\bfx,\bfc^*_j)+\sum_{\substack{r\in [m] :\bfc'_j[r]\neq \bfc_j^*[r]}}\vert \bfx[r]-\bfc'_j[r]\vert\right)\\
&=& \Delta_2+\sum_{j\in [\ell]\setminus I'}\sum_{\bfx\in X_j}\sum_{\substack{r\in [m] :\bfc'_j[r]\neq \bfc_j^*[r]}}\vert \bfx[r]-\bfc'_j[r]\vert \qquad (\mbox{by \eqref{eqn0O2}})\\
&=& \Delta_2+\sum_{j\in [\ell]\setminus I'}\sum_{\substack{r\in [m] :\bfc'_j[r]\neq \bfc_j^*[r]}}\sum_{\bfx\in X_j} \vert \bfx[r]-\bfc'_j[r]\vert \\
&\leq & \Delta_2+\sum_{j\in [\ell]\setminus I'}\sum_{\substack{r\in [m] :\bfc'_j[r]\neq \bfc_j^*[r]}} \vert X_j\vert  \\
&\leq & \Delta_2+\sum_{j\in [\ell]\setminus I'} \frac{3\alpha}{k}\cdot \Delta_1 \qquad\qquad\qquad\qquad\qquad\qquad(\mbox{by Claim}~\ref{claim:aux1})\\
&\leq & \Delta_2+({3\alpha}\cdot \Delta_1).
\end{eqnarray*}
The last inequality completes the proof of Claim~\ref{claim:secpart}, which in turn completes the proof of the lemma.
\end{proof}
\end{proof}

\section{Putting all together:  Proof of Theorem~\ref{thm:mainrclustering}}

As we already mentioned in Section~\ref{sec:non-reduceb}, the most essential part of the proof of Theorem~\ref{thm:mainrclustering} is to approximate irreducible instances. We start from the approximation algorithm for irreducible instances and then use this algorithm to prove Theorem~\ref{thm:mainrclustering}.

\subsection{Approximating irreducible instances}

Now we have sufficient ingredients to design an algorithm for outputting a $(1+\varepsilon)$-approximate solution under the assumption that the given instance is $(k,\frac{\varepsilon}{8k^*})$-irreducible for some $k^*\geq k$.  
More precisely, in this section we prove the following theorem.

\begin{theorem}
\label{thm:iterm}
There is an algorithm with the following specifications. 
\begin{itemize}
\item The input of the algorithm is an instance $J=(X,k,{\cal R}=\{R_1,\ldots,R_m\})$ of \rclustering, $\varepsilon>0$, and $k^*\geq k$
such that $J$ is $(k,\frac{\varepsilon}{8k^*})$-irreducible. 
\item The output of the algorithm is a solution $C$ to $J$ such that $\Hdist(X,C)\leq (1+\frac{\varepsilon}{40k^*})\opt(J)$ with probability at least 
\[p(k^*,k,\varepsilon)= 
\left( \frac{\varepsilon}{1+\varepsilon} \right)^k \cdot \left( \frac{\varepsilon^k }{\left(40k\cdot k^*\cdot (5^k-1)\right)^{k-1}\cdot 2k\cdot (40k^*+\varepsilon)} \right)^{\frac{c'\cdot k^2}{\varepsilon^2} \log \frac{1}{\varepsilon}}.
\]
\item The running time of the algorithm is  $ 2^{\OO(k^2+k\log(k^*))}\cdot n\cdot m \cdot \left(\frac{1}{\varepsilon} \right)^{\OO(k)}$.
\end{itemize}
\end{theorem}

First we provide an overview of the algorithm informally without mentioning actual values of the parameters. Then we reason about how to set different parameters which will lead to the required algorithm. 
And then we give a formal proof of   Theorem~\ref{thm:iterm}.

On  a high level our algorithm works as follows. Let $J=(X,k,{\cal R})$ be the input instance. We consider a tuple $(C',S,\delta)$ to be a partial solution where $S\subseteq X$ and $C'\subseteq \{0,1\}^m, \vert C'\vert \leq k$ and $\delta\geq 0$. 
We say that a partial solution $(C',S,\delta)$ is a {\em good} partial solution, if $(C',X\setminus S)$ is $\delta$-extendable for $J$.  That is $C'$ is a set of cluster centers.
Moreover $C'$ can be extended to a $(1+\delta)$-approximate solution with $X\setminus S$ being assigned to clusters corresponding to $C'$. Initially we set our set of partial solution to be ${\cal P}=\{(\emptyset,X,0)\}$. Clearly $(\emptyset,X,0)$ is a good partial solution.  At each iteration we maintain an invariant that ${\cal P}$ contains a  good partial solution with high probability.
At any step if there is a  partial solution $(C_1,S,\delta)\in\PP$ such that $\vert C_1\vert<k$ and $S\neq \emptyset$ we do the following. First we delete $(C_1,S,\delta)$ from $\PP$ and then we add following partial solutions to $\PP$. Let $\bfv_1,\ldots,\bfv_{\vert S\vert}$ be the vectors in $S$ ordered according to the non-decreasing order of $\hdist(\bfv_i,C_1)$.

\begin{itemize}
\item[$(i)$] We add $(C_1,S',\delta)$, where $S'$ is the last $\lceil \frac{\vert S\vert}{2} \rceil$ vectors in the order $\bfv_1,\ldots,\bfv_{\vert S\vert}$.
\item[$(ii)$] For any choice of $k'\in \{1,\ldots,k-\vert C_1\vert \}$ we  extend the set $C_1$  assuming $(C_1,X\setminus S)$ is $\delta$-extendable (see the below paragraph for a brief explanation).
\end{itemize}  

Assume that $C_2$ is a good $\delta$-extension of $(C_1,X\setminus S)$ (see Definition~\ref{def:goodetn}). 
Since $C_2$ is a good $\delta$-extension of $(C_1,X\setminus S)$, there exist $I\in \binom{[k]}{\vert C_1\vert}$  and a partition $Z_1\uplus Z_2$ of $S$ such that $C_2$ is an optimum solution to $J'=(Z_2,\ell=k-\vert C_1\vert,\RR(I,C_1))$ 
and $\Hdist(S,C_1\cup C_2)=\Hdist(Z_1,C_1)+\Hdist(Z_2,C_2)$ (see Lemma~\ref{lem:bestC2}). 
Let $B'$ be the set defined in Lemma~\ref{lem:safetoremove}. If $\vert B'\vert \geq \lfloor \frac{\vert S \vert}{2}\rfloor$, then by Lemma~\ref{lem:safetoremove}$(ii)$ and Observation~\ref{obs:largeext}, $(C_1,S',\delta)$ is $\delta$-extendable for $J$ (this is covered in 
item $(i)$ above). Suppose the size of $B'$ is at most $\lfloor\frac{\vert S\vert}{2}\rfloor$.  
Then, by Lemma~\ref{lem:safetoremove}$(iii)$, we know that $\vert Z_2\vert$ is a constant fraction of  $\vert S \vert$. 
Therefore  
the size of the largest cluster among the clusters correspond to $C_2$ in the instance $J'$ is also a constant fraction of $\vert S\vert$. 
Let $X_1,\ldots,X_{\ell}$ be the clusters corresponding to the solution $C_2$, where $\ell=\vert C_2\vert$. 
For the ease of presentation assume that $\vert X_1\vert \geq \vert X_2\vert \ldots \geq \vert X_{\ell}\vert$. 
Let $k'$ be the smallest integer such that $\vert X_{k'}\vert$ 
is ``much larger than''  $\vert X_{k'+1}\vert$. 
In other words, for any $i\in \{1,\ldots,k'-1\}$, the size of $X_i$ 
is at most a constant times $\vert X_{i+1}\vert$. 
This implies that the size of $X_j$ is a constant fraction of $\vert S\vert$
for any $j\in \{1,\ldots,k'\}$. Therefore we use  Lemma~\ref{lem:samplelargecluster} 
to compute an approximate set of centers $C'$ for the clusters $X_1,\ldots,X_j$. 
Then we add $(C_1\cup C',S,5\delta+4\alpha)$ to ${\cal P}$ (this case is covered in item $(ii)$ above). 
The explanation for setting the parameter $\alpha$ 
and how to estimate weights $w_i$ used in Lemma~\ref{lem:samplelargecluster}
is given in the next paragraph.   
%
At the end, for any tuple $(C,S,\delta)\in \PP$, either $\vert C\vert =k$ or $S=\emptyset$. 
Then we output the best solution among all the tuples in $\PP$. The correctness of the algorithm will follow from the  invariant that $\PP$ contains a good partial solution with high probability at each iteration.

Now we explain how to set different parameters for the algorithm, which are used in Lemmas~\ref{lem:samplelargecluster}, \ref{lem:safetoremove}, and  \ref{lem:findmore}. We have an assumption that $J$ is $(k,\frac{\varepsilon}{8k^*})$-irreducible. In Lemmas~\ref{lem:safetoremove} and \ref{lem:findmore}, we want $J$ to be   $(k,5\delta')$-irreducible. So we set $\delta'=\frac{\varepsilon}{40k^*}$.  Initially we set $\PP=\{(\emptyset,X,0)\}$. At each iteration we extend a partial solution $(C_1,S,\delta_1)$ either by deleting half of vectors from $S$ or by computing a set of center vectors using  Lemma~\ref{lem:samplelargecluster} and the correctness of the step (assuming 
$(C_1,X\setminus S)$ is $\delta_1$-extendable) follows from  Lemma~\ref{lem:findmore}. 
For the initial application of Lemma~\ref{lem:findmore} (i.e., where $(C_1,S,\delta_1)=(\emptyset,X,0)$ ), we have that $\delta_1=0$. 
Then after each application  of Lemma~\ref{lem:findmore} we get a partial solution which we expect to be $(5\delta_1+4\alpha)$-extendable (we fix $\alpha$ later). The number of times Lemma~\ref{lem:findmore} is applied to get a particular solution in ${\cal P}$ (at the end of the algorithm) is at most $k$. This implies that at the end some solution in ${\cal P}$ is $\gamma(k)$-extendable (by Observation~\ref{obs:largeext}), where $\gamma(k)$ can be obtained from the following recurrence relation. 
\begin{eqnarray*}
 &&\gamma(0) =0,    \mbox{ and }\\
 &&\gamma(k') =5\gamma(k'-1)+4\alpha,   \mbox{ for } k'\in \{1,2,\ldots,k\}.
\end{eqnarray*}
The above recurrence relation solves to $\gamma(k')=(5^{k'}-1)\alpha$ for any $k'\in \{0,1,\ldots,k\}$.  Moreover, by Lemma~\ref{lem:findmore}, we need 
$\gamma(k')$ to be at most $\delta'$ for all $k'\in \{0,1,\ldots,k\}$. Thus we set $\alpha=\frac{\delta'}{(5^k-1)}$. 
From Lemma~\ref{lem:samplelargecluster}, we expect 
a $(1+\alpha)$-approximate solution for a restricted instance derived from $J$. 
So we use $\varepsilon=\frac{\alpha}{7}$ and $\delta=\frac{\alpha}{7}$ in the application 
of Lemma~\ref{lem:samplelargecluster}. To apply Lemma~\ref{lem:samplelargecluster}, we also need  the value for $\beta$ and promised bounds on the  sizes of the clusters for which we are seeking center vectors. Towards that we give a detailed explanation of how to use 
 Lemma~\ref{lem:samplelargecluster}. Notice that   $(C_1,S,\delta_1)$ is a partial solution already computed and we assumed 
 that $(C_1,X\setminus S)$ is $\delta_1$-extendable for $J$. Let $C_2$ is a good $\delta$-extension of $(C_1,X\setminus S)$.  
 Since $C_2$ is a good $\delta$-extension of $(C_1,X\setminus S)$, there exist $I\in \binom{[k]}{\vert C_1\vert}$  
and a partition $Z_1\uplus Z_2$ of $S$ such that $C_2$ is an optimum solution to $J'=(Z_2,\ell=k-\vert C_1\vert,\RR(I,C_1))$ 
and $\Hdist(S,C_1\cup C_2)=\Hdist(Z_1,C_1)+\Hdist(Z_2,C_2)$ (see Lemma~\ref{lem:bestC2}). 
Let $B'$ be the set defined in Lemma~\ref{lem:safetoremove}. 
 The application of Lemma~\ref{lem:samplelargecluster} is important for the partial solution 
 $(C_1,S,\delta_1)$ only when the size of $B'$ is at most $\frac{\vert S\vert}{2}$. 
Then, by Lemma~\ref{lem:safetoremove}$(iii)$, we know that $\vert Z_2\vert \geq \left(\frac{\delta'}{2(1+\delta')}\right)\vert S \vert$. 
Therefore, the largest cluster corresponding to the solution $C_2$ of $J'$ is at least $\left(\frac{\delta'}{2k(1+\delta')}\right)\vert S\vert$. 
Let $X_1,\ldots,X_{\ell}$ be the clusters corresponding $C_2$, where $\ell=\vert C_2\vert$. For the ease of presentation assume that $\vert X_1\vert \geq \vert X_2\vert \ldots \geq \vert X_{\ell}\vert$. Let $k'$ be the smallest integer such that $\vert X_{k'}\vert > \frac{k}{\alpha} \vert X_{k'+1}\vert$. Therefore we have that $\vert X_1\vert \leq \frac{k}{\alpha}\vert X_2\vert \leq \ldots \leq \left(\frac{k}{\alpha}\right)^{k'-1}\vert X_{k'}\vert$ (Here the set $[k']$ plays the role of $I'$ in Lemma~\ref{lem:findmore}, i.e., $[k']$ is a set of indices of $\frac{k}{\alpha}$-heavy clusters). 
This implies that the size of $X_j$ is at least $\left(\frac{\alpha}{k}\right)^{k'-1}\vert X_1\vert\geq \left(\frac{\alpha}{k}\right)^{k-1}\frac{\delta'}{2k(1+\delta')} \vert S\vert$ (because $\frac{\alpha}{k}\leq 1$) for any $j\in [k']$. So we set  $\beta=\left(\frac{\alpha}{k}\right)^{k-1}\frac{\delta'}{k(2+\delta')}$. We set the value $c$ in Lemma~\ref{lem:samplelargecluster} to be $\vert X_{k'}\vert$. Notice that we do not have to know the value of $c$ explicitly, but instead we need to know the weights $w_i$s within a  promised bound.  From the value of $c$ it is clear that  for any $j\in [k']$, $\frac{\vert X_j\vert}{c}\leq \left(\frac{k}{\alpha}\right)^{k'-j}\leq \left(\frac{k}{\alpha}\right)^{k'-1}$. 
Let $h=\left(\frac{k}{\alpha}\right)^{k'-1}$. Then there exists $w_j\in \{(1+\delta)^0,(1+\delta)^1\ldots,(1+\delta)^{\log_{1+\delta} h}\}$ such that $\frac{\vert X_j\vert}{c}\leq w_j \leq \frac{(1+\delta)\vert X_j\vert}{c}$. Thus we apply  Lemma~\ref{lem:samplelargecluster} 
for all possible values $w_1,\ldots,w_{k'}\in \{(1+\delta)^0,(1+\delta)^1\ldots,(1+\delta)^{\log_{1+\delta} h}\}$ and extend 
the set $C_1$ in each case. 

Now we are ready to give the formal proof of the theorem. 
%
%

\begin{algorithm}[t]
\KwIn{An instance $J=(X,k,{\cal R}=\{R_1,\ldots,R_m\})$ of \rclustering\ and $\varepsilon>0$.}
\KwOut{A solution $C\subseteq \{0,1\}^m$ for $J$.}
$\PP \gets\{(\emptyset,X,0)\}$  \\
$\delta' \gets\frac{\varepsilon}{40k^*},\; \alpha\gets \frac{\delta'}{(5^k-1)}$, 
$\beta \gets \left(\frac{\alpha}{k}\right)^{k-1}\cdot \frac{\delta' }{2k(1+\delta')}$, $\varepsilon\gets \frac{\alpha}{7}$, and $\delta\gets \frac{\alpha}{7}$. \\
\For{$(C_1,S,\delta_1) \in \PP$ such that $\vert C_1\vert <k$ and $S\neq \emptyset$ \label{iteration} }{
$\PP\gets \PP\setminus \{(C_1,S,\delta_1)\}$ \label{line:addition0} \\
Let $\pi$ be a linear order of $S$ according to the non-decreasing distance $\hdist(\bfv,C_1)$, where $\bfv\in S$.\label{step:sort} \\
Let $S'$ be the last $\lceil \frac{\vert S \vert}{2}\rceil$ vectors in the order $\pi$. \label{line:Sprime} \\
$\PP\gets \PP\cup \{(C_1,S',\delta_1)\}$ \label{line:addition}\\
Guess  $I\in \binom{[k]}{\vert C_1\vert}$ such that $(i)$ $<C_1, \proj{\RR}{I}>$ (for a proper ordering of the vectors in $C_1$),  $(ii)$ there is a good $\delta_1$-extension  $C_2$ of $(C_1,X\setminus S)$ and  $(iii)$ $<C_2, \RR(I,C_1)>$.\label{step:guesI}\\ 
$\RR' \gets \RR(I,C_1)$.\\
$\ell\gets k-\vert C_1\vert$\\
\For{$k'\in \{1,\ldots, k-\vert C_1\vert \}$ and $I'\in \binom{[\ell]}{k'}$ \label{step:linekprime}}{
$h\gets \left(\frac{k}{\alpha}\right)^{k'-1}$ \\
\For{ $w_1,\ldots,w_{k'} \in \{(1+\delta)^0,(1+\delta)^1\ldots,(1+\delta)^{\log_{1+\delta} h}\} $ \label{step:linekprime1}}{
$C'\gets$ the output of Algorithm $\A$ from  Lemma~\ref{lem:samplelargecluster} on input $S,k',\proj{\RR'}{I'}$, $\beta$, $\delta$, ${\varepsilon}$,  and $w_1,\ldots,w_{k'}$ \label{line:Cprime} \\
$\PP\gets \PP\cup \{(C_1\cup C',S,5\delta_1+4\alpha)\}$ \label{line:addition}
}
}
}
Let $C$ be such that there is $(C,S^*,\delta^*)\in \PP$ for some $S^*\subseteq X, \delta^*\geq 0$ and $\Hdist(X,C)=\min \{\Hdist(X,C')\colon \mbox{there is $S'\subseteq X,\delta_1\geq 0$ such that }(C',S',\delta_1)\in \PP\}$.\\
\KwRet a solution $D\supseteq C$ using Proposition~\ref{prop:smallsolution}.
\caption{Algorithm for \rclustering\ assuming the input instance is $(k,\frac{\varepsilon}{8k^*})$-irreducible, where $k$ is the number of vectors allowed in the solution, $\varepsilon$ is the approximation factor and $k^*\geq k$.}
 \label{alg:imp}
\end{algorithm}

\begin{proof}[Proof of Theorem~\ref{thm:iterm}]
The pseudocode  of our algorithm is given in Algorithm~\ref{alg:imp}. 
First, we define an {\em iteration} of the algorithm to be the execution of one step of the {\bf for loop} at Line~\ref{iteration}. That is at the beginning of an iteration one partial solution will be deleted from $\PP$ and later during the execution of the iteration many partial solutions  will  be added to $\PP$ (see Lines~\ref{line:addition0} and \ref{line:addition}). 
Next, we prove that the algorithm terminates.  Notice that when there is no partial solution $(C_1,S,\delta_1)\in \PP$ with $\vert C_1\vert < k$ and $S\neq \emptyset$, then the algorithm terminates. When there is a partial solution $(C_1,S,\delta_1)\in \PP$ with $\vert C_1\vert < k$ and $S\neq \emptyset$, then there is an iteration of the algorithm (Line~\ref{iteration}) where in the {\bf for} loop we consider $(C_1,S,\delta_1)$, delete $(C_1,S,\delta_1)$ from $\PP$ and add many partial solutions to $\PP$. For each such partial solution $(C'_1,S',\delta_2)$, either $\vert C_1'\vert > \vert C_1\vert$ or  $\vert S'\vert < \vert S \vert$. This implies that Algorithm~\ref{alg:imp} will terminate.

Now we prove the correctness of the algorithm. That is we prove that   Algorithm~\ref{alg:imp} will return a 
$(1+\frac{\varepsilon}{40k^*})$-approximate solution with probability at least  $p(k^*,k,\varepsilon)$.

For convenience we let the $0{th}$ iteration to be the initial assignments before any of the execution of Line~\ref{iteration}.  That is at the end of the $0th$ iteration we have that $\PP=\{(\emptyset,X,0)\}$. When we apply in line 14 
Algorithm $\A$ from  Lemma~\ref{lem:samplelargecluster} on input $S,k',\proj{\RR'}{I'}$, $\beta$, $\delta$, ${\varepsilon}$,  and $w_1,\ldots,w_{k'}$,  we know that  by Lemma~\ref{lem:samplelargecluster},  it outputs a solution with probability at least
$\frac{\varepsilon\cdot \beta^{r\cdot k}}{1+\varepsilon}$, where $r= \Theta(\frac{k}{\varepsilon^2} \log \frac{1}{\varepsilon})$.
 Let $c'$ be a constant such that   $r=c'\cdot \frac{k}{\varepsilon^2} \log \frac{1}{\varepsilon}$.
\begin{quote}
{\bf Correctness Invariant:}
At the end of every iteration of the algorithm, there is a partial solution $(C,S,\delta_1)\in \PP$ such that $(C,X\setminus S)$ is $\delta_1$-extendable 
for $J$ with probability at least 
$\left(\frac{\varepsilon \beta^r}{1+\varepsilon}\right)^q$
where $q=\vert C\vert$. 
\end{quote}

We need the following observation to prove the correctness invariant. 

\begin{observation}
\label{obs:deltaprime}
At any step of  the algorithm,  for any partial solution $(C,S,\delta_1)\in \PP$, 
$\delta_1\leq \delta'$. 
\end{observation}
\begin{proof}[Proof Sketch]
By induction on the number of iteration of the algorithm one can prove that at the end of each iteration, 
for any partial solution $(C,S,\delta_1)\in \PP$, $\delta_1 \leq \gamma(\vert C\vert)$. 
Recall that $\gamma(k')=(5^{k'}-1)\alpha$ and $\gamma(k')\leq \delta'$ for all $k'\in \{0,1,\ldots,k\}$.
Thus the observation follows. 
\end{proof}

\begin{claim}
\label{claim:corr}
Correctness invariant is maintained at the end of every iteration.  
\end{claim}
\begin{proof}
We prove the claim using induction on the number of iterations. The base case is for the $0th$ iteration. Since 
at the end of $0th$ iteration $(\emptyset,X,0)\in \PP$ and the fact that  $(\emptyset,\emptyset)$ is $0$-extendable for $J$ with 
probability $1$, the base case follows. 
 Now we consider the induction 
step. That is consider the iteration $i$ of the algorithm where $i>0$. Let $(C_1,S,\delta^*)$ be the partial solution for which the iteration $i$ corresponds to. At the beginning of iteration  $i$, $(C_1,S,\delta^*)\in \PP$ and by induction hypothesis there is a partial solution $(C',S',\delta_1)\in \PP$ satisfying 
the properties mentioned in the correctness invariant. Notice that during the iteration $i$, we will delete $(C_1,S,\delta^*)$ from 
$\PP$ and add many partial solutions to $\PP$. If $(C',S',\delta_1)\neq (C_1,S,\delta^*)$ then at the end of iteration $i$, $(C',S',\delta_1)\in \PP$ and the invariant follows.

So now we assume that $(C_1,S,\delta^*)=(C',S',\delta_1)$. This implies that $(C_1,X\setminus S)$ is $\delta^*$-extendable for $J$ with probability at least  $\left(\frac{\varepsilon \beta^r}{1+\varepsilon}\right)^q$,  where $q=\vert C_1\vert$. 
By Lemma~\ref{lem:bestC2}, there is 
a good $\delta^*$-extension $C_2$ of $(C_1,X\setminus S)$, a 
partition $Z_1\uplus Z_2$ of $X\setminus S$ and $I\in \binom{k}{\vert C_1\vert}$ such that  $C_1\in \proj{\RR}{I}$ and  $C_2$ is an optimum solution to $(Z_2,k-\vert C_1\vert, \RR(I,C_1))$
and $\Hdist(S,C_1\cup C_2)=\Hdist(Z_1,C_1)+\Hdist(Z_2,C_2)$. 
 Let $t=\min_{\bfc\in C_1,\bfc'\in C_2}\hdist(\bfc,\bfc')$.
Let $B'=\bigcup_{\bfc\in C_1}\BB(\bfc,t'/2)\cap S$ and $S^*=S\setminus B'$. Then by Lemma~\ref{lem:safetoremove}, we get the following.
\begin{itemize}
\item[$(i)$]$B'\subseteq Z_1$ and $B'$ consists of the first $\vert B'\vert$ vectors of $S$ in the ordering  
according to the non-decreasing distance $\hdist(\bfx,C_1)$ where $\bfx\in S$. 
\item[$(ii)$] $\Hdist(S^*,C_1\cup C_2)=\Hdist(Z_1\setminus B',C_1)+\Hdist(Z_2,C_2)$. Moreover, $(C_1, X\setminus S^*)$ is $\delta^*$-extendable for $J$ and $C_2$ is a $\delta^*$-extension of $(C_1, X\setminus S^*)$.   
\item[$(iii)$] Since $\delta^* \leq \delta'$ (by Observation~\ref{obs:deltaprime})  and the fact that $J$ is $(k,5\delta')$-irreducible (because of our assumption), we have that   $\vert Z_2\vert \geq \left(\frac{\delta'}{1+\delta'}\right) \vert S^*\vert$.  If $\vert B'\vert \leq \frac{\vert S \vert}{2}$, then 
$\vert Z_2\vert \geq \left(\frac{\delta'}{2(1+\delta')}\right) \vert S\vert$. 
\end{itemize}


Suppose  $\vert B'\vert > \frac{\vert S \vert}{2}$. Recall the definition of $S'$ from Line~\ref{line:Sprime}. Notice that $S\setminus S'\subseteq B'$. Therefore, by the assumption that $(C_1,X\setminus S)$ is $\delta^*$-extendable for $J$ and by 
Lemma~\ref{lem:safetoremove}$(ii)$ and Observation~\ref{obs:largeext},  $(C_1,X\setminus S')$ is  $\delta^*$-extendable for $J$. Since 
 $(C_1,X\setminus S)$ is $\delta^*$-extendable with probability at least  $\left(\frac{\varepsilon \beta^r}{1+\varepsilon}\right)^q$, $(C_1,X\setminus S')$ is $\delta^*$-extendable with probability at least  $\left(\frac{\varepsilon \beta^r}{1+\varepsilon}\right)^q$.

Now consider the case $\vert B'\vert \leq \frac{\vert S \vert}{2}$. We know that $C_2$ is an optimum solution to 
$(Z_2,\ell=k-\vert C_1\vert, \RR(I,C_1))$.  Let $C_2=\{\bfc_1,\ldots,\bfc_{\ell}\}$ and $Y_1,\ldots Y_{\ell}$ be the clusters corresponding to the solution $C_2$ of $(Z_2,\ell=k-\vert C_1\vert, \RR(I,C_1))$.  Let $\pi$ be a permutation of $[\ell]$ such that $\vert Y_{\pi(1)}\vert \geq \ldots \geq \vert Y_{\pi(\ell)}\vert$. 
Let $X_j=Y_{\pi(j)}$ and $\bfc_j^*=\bfc_{\pi(j)}$ for all $j\in [\ell]$.  Let $k'$ be the smallest integer in $[\ell]$ such that $\vert X_{k'}\vert >\left(\frac{k}{\alpha}\right)\vert X_{k'+1}\vert$. This implies that $\vert X_1 \vert \leq \left(\frac{k}{\alpha}\right)\vert X_2\vert \leq \ldots \leq \left(\frac{k}{\alpha}\right)^{k'-1}\vert X_{k'}\vert$. Thus, by the fact that  $\vert X_1\vert \geq \frac{\vert Z_2\vert}{\ell} \geq \frac{\vert Z_2\vert}{k}$, 
we have that for any $j\in [k']$,

\begin{equation}
\label{eqn:largecluster}
\vert X_{j}\vert \geq \left(\frac{\alpha}{k}\right)^{k'-1}\frac{\vert Z_2\vert}{k}.
\end{equation}
Let $I'=\{\pi(j)\colon j\in [k']\}$. By statement $(iii)$ above and \eqref{eqn:largecluster},
for all $j\in I'$, 
\begin{equation}
\label{eqn:choicebeta}
\vert Y_j\vert \geq  \left(\frac{\alpha}{k}\right)^{k'-1}\cdot \frac{\delta'}{2k(1+\delta')}\vert S\vert\geq \left(\frac{\alpha}{k}\right)^{k-1}\cdot \frac{\delta' \cdot \vert S\vert}{2k(1+\delta')} \qquad(\mbox{Since }\frac{\alpha}{k}\leq 1)
\end{equation}
This implies that for all $j\in I'$, $\vert Y_j\vert \geq \beta \vert S\vert$. 
Let $C^*=\{\bfc_j\colon j\in I'\}$. 
Let $V=\sum_{j\in I'} \Hdist(Y_j,\bfc_j)$. 
Moreover, $(a)$ $\{Y_j\colon j\in I'\}$ is a set of $\frac{k}{\alpha}$-heavy clusters (see Definition~\ref{def:heavyclusters}). 
Let $c=\vert Y_{k'}\vert$. Notice that $(b)$ for any $j\in I'$, $\frac{Y_j}{c}\leq  \left( \frac{k}{\alpha}\right)^{k'-1}$. 
Let $\RR'=\RR(I,C_1)$ and $J'$ be the instance $(Z=\bigcup_{j\in I'}Y_j,k',\proj{\RR'}{I'})$. 
Now in the iteration $i$, consider the execution of {\bf for loop} at Line~\ref{step:linekprime}  for  $k'$ and $I'$.  Notice that for any $j\in I'$, there exists $w \in \{(1+\delta)^0,(1+\delta)^1\ldots,(1+\delta)^{\log_{1+\delta} h}\}$, where $h=\left(\frac{k}{\alpha}\right)^{k'-1}$, 
such that $\frac{\vert Y_j\vert}{c}\leq w \leq \frac{(1+\delta)\vert Y_j\vert}{c}$. 
Now consider the {\bf for loop} at Line~\ref{step:linekprime1} for values $w_1,\ldots,w_{k'}$ such that  $\frac{\vert Y_{\pi(j)}\vert}{c}\leq w_j \leq \frac{(1+\delta)\vert Y_{\pi(j)}\vert}{c}$ for all $j\in [k']$.  Then, by Lemma~\ref{lem:samplelargecluster}, in Line~\ref{line:Cprime} we get a set $C'=\{\bfc_i'\colon i\in I'\}$ of $k'$ cluster centers such that with probability at least  $\frac{\varepsilon \cdot \beta^{r\cdot k'}}{1+\varepsilon}$,
\begin{equation}
\sum_{i\in I'}\Hdist(Y_i,\{\bfc_i'\})\leq (1+\varepsilon)^2 (1+\delta)V\leq (1+\alpha)\cdot \sum_{j\in I'} \Hdist(Y_j,\bfc_j). \label{boundcorr}
\end{equation}
Therefore, $(C_1\cup C',S,5\delta^*+4\alpha)$ belongs to $\PP$ at the end of iteration $i$ with probability 
at least $\left(\frac{\varepsilon \cdot \beta^{r}}{1+\varepsilon}\right)^{(q+ k')}=\left(\frac{\varepsilon \cdot \beta^{r }}{1+\varepsilon}\right)^{\vert C_1\cup C'\vert}$. 
Then by Lemma~\ref{lem:findmore}, $(C_1\cup C',X\setminus S)$ is $(5\delta^*+4\alpha)$-exdendable for $J$ and this completes 
the proof of the claim. 
\end{proof}

Now for the proof of the correctness of the algorithm consider the invariant at the end of the last iteration. 
At the end of last iteration for all $(C,S,\delta_1)\in \PP$, we have that either $\vert C\vert=k$ or $S=\emptyset$.  Then by Claim~\ref{claim:corr}, at the end of the last iteration, $\PP$ contains a partial solution $(C,S,\delta_1)$ which is $\delta_1$-extendable 
with probability at least 
\begin{equation}\label{eq:p-eps}
p=\left(\frac{\varepsilon \beta^r}{1+\varepsilon}\right)^{k}.
\end{equation} 
Then the output  set $D\supseteq C$ is  a $(1+\delta_1)$-approximate solution of $J$ with probability at least $p$, 
by Proposition~\ref{prop:smallsolution}. By Observation~\ref{obs:deltaprime},  $D$ is a   $(1+\frac{\varepsilon}{40k^*})$-approximate solution of $J$. 
By substituting the value of $\beta$ and $r$ into \eqref{eq:p-eps}, 
we bound the value of $p$ as follows.  
\begin{eqnarray*}
p&=&\left( \frac{\varepsilon}{1+\varepsilon} \right)^k \cdot \left(   \left( \frac{\delta'}{k(5^k-1)} \right)^{k-1} \frac{\delta'}{2k(1+\delta')}  \right)^{r\cdot k}\\
&=& \left( \frac{\varepsilon}{1+\varepsilon} \right)^k \cdot \left( \frac{\left(\frac{\varepsilon}{40k^*}\right)^k }{\left( k(5^k-1)\right)^{k-1}\cdot 2k\cdot (1+\frac{\varepsilon}{40k^*})} \right)^{r\cdot k}\\
&=& \left( \frac{\varepsilon}{1+\varepsilon} \right)^k \cdot \left( \frac{\varepsilon^k }{\left(40k\cdot k^*\cdot (5^k-1)\right)^{k-1}\cdot 2k\cdot (40k^*+\varepsilon)} \right)^{r\cdot k}\\
&=& \left( \frac{\varepsilon}{1+\varepsilon} \right)^k \cdot \left( \frac{\varepsilon^k }{\left(40k\cdot k^*\cdot (5^k-1)\right)^{k-1}\cdot 2k\cdot (40k^*+\varepsilon)} \right)^{\frac{c'\cdot k^2}{\varepsilon^2} \log \frac{1}{\varepsilon}}
\end{eqnarray*}

%
%
\paragraph*{Running time.}
Now we analyze the running time of the algorithm.  Notice that in the algorithm initially we have $\PP=\{(\emptyset,X,0)\}$ 
and in each step we delete a partial solution from $\PP$ and add many partial solutions. Towards analysing the running time 
we define a node-labelled rooted tree $T$ as follows. The root is labelled with  $(\emptyset,X,0)$. 
Each node of the  tree is labelled with a partial solution $(C,S,\delta_1)$, which corresponds an execution of {\bf for loop} at Line~\ref{iteration}. For a node labelled with $(C,S,\delta_1)$, let $S_1,\ldots,S_{\ell}$ be the tuples added to $\PP$ during the 
iteration corresponding to $(C,S,\delta_1)$. Then the node labelled $(C,S,\delta_1)$ will have $\ell$ children and they are labelled 
with $S_1,\ldots,S_{\ell}$, respectively.  Therefore, the number of iterations in the algorithm is equal to the number of nodes in the tree 
$T$. In each iteration (corresponding to a partial solution $(C,S,\delta_1)$) we sort the set $S$ of  vectors  according to the Hamming distance to $C$  (see Line~\ref{step:sort}). This can be done in time $\OO(kn'm)$, where $n'=\vert S\vert$.  Then we add $(C,S',\delta_1)$ to $\PP$, where 
$\vert S'\vert \leq \frac{\vert S \vert}{2}$.   Then because of  Lines~\ref{step:guesI}, \ref{step:linekprime} and \ref{step:linekprime1}, 
we add at most $L=2^{2k} \cdot \left(\log_{1+\delta}\frac{k}{\alpha}\right)^k$ tuples to $\PP$ (in Line~\ref{line:Cprime}) where 
the cardinality of the first entry of the each tuple is strictly more that $\vert C\vert$. 
The time required to execute the Line~\ref{line:Cprime} is at most $\OO\left(m 2^k \left(\frac{k}{\varepsilon}\right)^2 \log \frac{1}{\varepsilon}\right)$ (by Lemma~\ref{lem:samplelargecluster}).
Therefore  the time spend in one iteration of the algorithm is at most  
 $$\OO\left(L\cdot 2^{k} n' \cdot m \cdot  \left(\frac{k}{\varepsilon}\right)^2\log \frac{1}{\varepsilon}\right).$$

For any node $v$ labelled with  $(C,S,\delta_1)$, let $N(k-\vert C\vert ,\vert S\vert)$ be the time taken by the iterations which are labelled by the nodes of the subtree of $T$, rooted at the node $v$.   The value of  $N(k-\vert C\vert,\vert S\vert )$ can be upper bounded using the following recurrence formula. There is a constant $c$ such that  for any $0\leq k'\leq k$ and $0\leq n'\leq n$,   
\begin{eqnarray}
N(k',n') \leq
\left\{ \begin{array}{clc}
c & \mbox{if} &  k'=0 \mbox{ or } n'=0\\ 
N(k'+\frac{n'}{2})+L\cdot N(k'-1,n') +
\left( cL\cdot 2^{k} n' m  \left(\frac{k}{\varepsilon}\right)^2\log \frac{1}{\varepsilon}\right)& \mbox{if} &  k', n'>0
\end{array}\right. \label{eqn:runtime}
\end{eqnarray}


Clearly, the running time of the algorithm will be upper bounded by $N(k,n)$. We claim that $$N(k,n)\leq c \cdot (2L)^{k} 2^{3k^2}\cdot n\cdot m \cdot \left(\frac{k}{\varepsilon}\right)^2\log \frac{1}{\varepsilon}.$$ 
Towards that we prove   that for any $0\leq k'\leq k$ and $0\leq n'\leq n$,  $N(k',n')\leq  c\cdot  2^{k} \cdot L^{k'} 2^{3{k'}^2}\cdot n'\cdot m \left(\frac{k}{\varepsilon}\right)^2\log \frac{1}{\varepsilon}$ using induction. 
The base case is when $k'=0$ or $n'=0$ and it holds by \eqref{eqn:runtime}. By induction hypothesis we have that 
\begin{eqnarray*}
N(k',n')\leq \left( c\cdot 2^{k} L^{k'} 2^{3{k'}^2}\frac{n'm}{2}\left(\frac{k}{\varepsilon}\right)^2\log \frac{1}{\varepsilon}\right) &+& 
\left(
 L  \cdot {c} 2^{k}  L^{k'-1} 2^{3{(k'-1)}^2} n' m \left(\frac{k}{\varepsilon}\right)^2\log \frac{1}{\varepsilon}
\right) \\ 
&+& \left(cL\cdot 2^{k} n'\cdot m \left(\frac{k}{\varepsilon}\right)^2\log \frac{1}{\varepsilon}\right).
\end{eqnarray*}
To prove that $N(k',n')\leq c 2^{k} \cdot L^{k'} 2^{3{k'}^2}\cdot n'\cdot m \cdot \left(\frac{k}{\varepsilon}\right)^2\log \frac{1}{\varepsilon}$, it is enough to prove 
that $  2^{3{k'}^2-1} + 2^{3{(k'-1)}^2} + 1 \leq 2^{3{k'}^2}$, which is true for any $k'\geq 1$. 

Therefore  we upper bound the running time of the algorithm as follow. 

\begin{eqnarray*}
N(k,n)&\leq& \left({c} 2^{k}  2^{3k^2}\cdot n\cdot m \left( \frac{k}{\varepsilon}\right)^2\log \frac{1}{\varepsilon}\right)\cdot L^k \\
&\leq& \left(2^{\OO(k^2)}\cdot n\cdot m \left( \frac{1}{\varepsilon}\right)^2\log \frac{1}{\varepsilon}\right)\cdot 2^{2k} \cdot \left(\log_{1+\delta}\frac{k}{\alpha}\right)^k \\
&\leq& \left(2^{\OO(k^2)}\cdot n\cdot m \left( \frac{1}{\varepsilon}\right)^2\log \frac{1}{\varepsilon}\right)\cdot \left(\frac{\ln \frac{k}{\alpha}}{\ln (1+\delta)} \right)^k \\
&\leq& \left( 2^{\OO(k^2+k\log(k^*))}\cdot n\cdot m \cdot \frac{1}{\varepsilon^2}\log \frac{1}{\varepsilon}\right)\cdot \left(\frac{\ln \frac{1}{\varepsilon}}{\ln (1+\delta)} \right)^k \\
&\leq& \left( 2^{\OO(k^2+k\log(k^*))}\cdot n\cdot m \cdot \frac{1}{\varepsilon^2}\log \frac{1}{\varepsilon}\right)\cdot \left(\frac{(1+\delta)\ln \frac{1}{\varepsilon}}{\delta} \right)^k\\
&\leq& \left( 2^{\OO(k^2+k\log(k^*))}\cdot n\cdot m \right)\cdot \left(\frac{1}{\varepsilon} \right)^{\OO(k)}.
\end{eqnarray*}
This completes the proof of the theorem. 
\end{proof}

\subsection{The final step}
Now we are ready to prove    Theorem~\ref{thm:mainrclustering}. 

Let $J=(X,k,\RR)$ be the input instance. 
Let $$p(k^*,k,\varepsilon)=
\left( \frac{\varepsilon}{1+\varepsilon} \right)^k \cdot \left( \frac{\varepsilon^k }{\left(40k\cdot k^*\cdot (5^k-1)\right)^{k-1}\cdot 2k\cdot (40k^*+\varepsilon)} \right)^{\frac{c'\cdot k^2}{\varepsilon^2} \log \frac{1}{\varepsilon}}$$
be the success probability from Theorem~\ref{thm:iterm}, where $c'$ is a constant. 
For each $I \subseteq [k]$, we apply Theorem~\ref{thm:iterm} on $(X,\vert I \vert ,\proj{\RR}{I})$ and $\frac{\varepsilon}{4}$ (where we substitute $k=\vert I \vert$ and $k^*=k$) $\frac{1}{p(k,\vert I\vert,\frac{\varepsilon}{4})}$ times. Let $C'_{I}$ be the best solution obtained by the above process. By using Proposition~\ref{prop:smallsolution}, 
let $C_{I}$ be the solution for $J$ obtained from $C'_I$. Then we output the best solution  among $\{C_I\colon I\subseteq [k]\}$. 
The running time of the algorithm mentioned in Theorem~\ref{thm:iterm} is $ 2^{\OO(k^2+k\log(k^*))}\cdot n\cdot m \cdot \left(\frac{1}{\varepsilon} \right)^{\OO(k)}$  and the running time of the algorithm mentioned in Proposition~\ref{prop:smallsolution} is linear in the input size. Thus the running time of our algorithm is at most $\frac{1}{p(k,k,\frac{\varepsilon}{4})} \cdot 2^{\OO(k^2)}\cdot \left(\frac{1}{\varepsilon} \right)^{\OO(k)}\cdot n\cdot m$. The value of $\frac{1}{p(k,k,\frac{\varepsilon}{4})}$ is upper bounded as follows. 

\begin{eqnarray*}
\frac{1}{p(k,k,\frac{\varepsilon}{4})}&=&\left( \frac{4+\varepsilon}{\varepsilon} \right)^k \cdot \left( \frac{\left(40 k^2\cdot (5^k-1)\right)^{k-1}\cdot 2k\cdot (40k+\frac{\varepsilon}{4})}{\left(\frac{\varepsilon}{4}\right)^k } \right)^{\frac{c'\cdot k^2}{\varepsilon^2} \log \frac{1}{\varepsilon}}\\
&=&\left(\frac{1}{\varepsilon}\right)^{\OO(k)}\cdot 2^{\OO\left(\frac{ k^4}{\varepsilon^2} \log \frac{1}{\varepsilon}\right)}\cdot \left(\frac{1}{\varepsilon}\right)^{\OO\left(\frac{ k}{\varepsilon^2} \log \frac{1}{\varepsilon}\right)}\\
&=&2^{\OO\left(\frac{ k^4}{\varepsilon^2} \log \frac{1}{\varepsilon}\right)}\cdot \left(\frac{1}{\varepsilon}\right)^{\OO\left(\frac{ k}{\varepsilon^2} \log \frac{1}{\varepsilon}\right)}
\end{eqnarray*}
Thus the running time of our algorithm is $2^{\OO\left(\frac{ k^4}{\varepsilon^2} \log \frac{1}{\varepsilon}\right)}\cdot \left(\frac{1}{\varepsilon}\right)^{\OO\left(\frac{ k}{\varepsilon^2} \log \frac{1}{\varepsilon}\right)}n \cdot m$. 


Now we prove the correctness of the algorithm. Let $\widehat{k}\in [k]$ be the integer defined in Lemma~\ref{lem:irr}. We know that 
$\opt_{\widehat k}(J)\leq (1+\frac{\varepsilon}{4})\opt(J)$. This implies that there is $I\in \binom{[k]}{\widehat k}$ such that 
the $\opt(J')$ is at most $(1+\frac{\varepsilon}{4})\opt(J)$, where $J'=(X,\widehat{k},\proj{\RR}{I})$. In the iteration of our algorithm corresponding to $I$, we get a solution $C_I'$ to $J'$ of cost at most $(1+\frac{\varepsilon}{160k})\opt(J')$ with probability at least $1-(1-p(k^*,\vert I\vert,\frac{\varepsilon}{4})^{1/p(k^*,\vert I\vert,\frac{\varepsilon}{4})}\geq 1-\frac{1}{e}$. By  Proposition~\ref{prop:smallsolution}, $C_I$ is a solution to $J$ of cost at most
$(1+\frac{\varepsilon}{160k})\opt(J')$. Thus by Lemma~\ref{lem:irr}, the cost of $C_I$  is $(1+\frac{\varepsilon}{4})(1+\frac{\varepsilon}{160k})\opt(J)\leq (1+\varepsilon)\opt(J)$. This completes the proof of the theorem.  


\bibliographystyle{siam}
\bibliography{book_pc}

\end{document}